\documentclass[11pt]{amsart}  

\usepackage{euler}
\usepackage{amscd} 
\usepackage{mathrsfs}
\usepackage{amsmath,amsbsy,amsthm,amsfonts,amssymb}
\usepackage{hyperref}
\newtheorem{theorem}{Theorem}[section]
\newtheorem{lemma}[theorem]{Lemma}
\newtheorem{proposition}[theorem]{Proposition}
\newtheorem{corollary}[theorem]{Corollary} 
\theoremstyle{definition}

\newtheorem{remark}[theorem]{Remark} 
\numberwithin{equation}{section}

\usepackage{color}
\newcommand{\BLUE}[1]{{\color{blue} \bf #1}}

\def\FF{\mathscr{F}}
\def\psiht{\widehat{\psi}^{\hbar}}
\def\psih{\psi^{\hbar}}
\def\phih{\phi^{\hbar}}
\def\upsilonh{\Upsilon^{\hbar}}
\def\SA{\Szero}

\def\HE{\theta}
\def\vfi{\varphi}
\def\ac{\beta}

\def\z*{\bar z}

\def\S{\mathcal S}
\def\L{\mathsf L}
\def\uno{\mathsf 1}
\def\H{\mathsf H}

\def\C{\mathcal C}

\def\dom{\text{\rm dom}}
\def\ran{\text{\rm ran}}

\def\FF{\mathscr F}

\def\RE{\mathbb R}
\def\CO{{\mathbb C}}
\def\o{$\bar{\text{\rm o}}$}

\def\dd{\displaystyle}

\def\ph*{\phi_\star}

\def\be{\begin{equation}}
\def\ee{\end{equation}}
\def\min{{\rm min}}
\def\max{{\rm max}}

\def\-{{\rm in}}
\def\+{{\rm ex}}
\newcommand{\sgn}{\operatorname{sgn}}

\def\ve{\lambda}
\def\Im{\operatorname{Im}}
\def\Re{\operatorname{Re}}
\def\sigmap{\breve{\sigma}}
\def\sigmaq{\sigma}
\def\sigmazero{\sigma_0}
\def\Szero{A}

\title[The semi-classical limit with delta potentials]{The semi-classical limit with delta potentials}
\author{Claudio Cacciapuoti}
\address{DiSAT, Sezione di Matematica, Universit\`a dell'Insubria, via Valleggio 11, I-22100
Como, Italy}
\email{claudio.cacciapuoti@uninsubria.it}

\author{Davide Fermi}
\address{Dipartimento di Matematica, Universit\`a di Milano, Via Cesare Saldini 50, I-20133 Milano, Italy 
and Istituto Nazionale di Fisica Nucleare, Sezione di Milano, I-20133 Milano, Italy}
\email{davide.fermi@unimi.it}

\author{Andrea Posilicano}
\address{DiSAT, Sezione di Matematica, Universit\`a dell'Insubria, via Valleggio 11, I-22100
Como, Italy}
\email{andrea.posilicano@uninsubria.it}

\thanks{The authors acknowledge the support of the National Group of Mathematical Physics (GNFM-INdAM)}

\begin{document}

\begin{abstract} We consider the semi-classical limit of the quantum evolution of  Gaussian coherent states whenever the Hamiltonian $\mathsf H$ is given, as sum of quadratic forms, by $\mathsf H= -\frac{\hbar^{2}}{2m}\,\frac{d^{2}\,}{dx^{2}}\,\dot{+}\,\alpha\delta_{0}$, with $\alpha\in\mathbb R$ and $\delta_{0}$ the Dirac delta-distribution at $x=0$. We show that the quantum evolution can be approximated, uniformly for any time away from the collision time and with an error of order $\hbar^{3/2-\lambda}$, $0\!<\!\lambda\!<\!3/2$, by the quasi-classical evolution generated by a self-adjoint extension of the restriction to $\mathcal C^{\infty}_{c}({\mathscr M}_{0})$, ${\mathscr M}_{0}:=\{(q,p)\!\in\!\mathbb R^{2}\,|\,q\!\not=\!0\}$, of ($-i$ times) the generator of the free classical dynamics; such a self-adjoint extension does not correspond to the classical dynamics describing the complete reflection due to the infinite barrier.  Similar approximation results are also provided for the wave and scattering operators.
\end{abstract}

\maketitle

\begin{footnotesize}
\emph{Keywords: Semiclassical dynamics; delta-interactions; coherent states; scattering theory.} 

\emph{MSC 2010:
81Q20; 
81Q10; 
47A40.
}  
\end{footnotesize}
 
\vspace{1cm}

\section{Introduction}
The semi-classical limit of Quantum Mechanics for Gaussian coherent states is a well established subject and, in the case of regular potentials, many comprehensive results are available (see, e.g., \cite{CR, Hag1, MF, 
Rob} and references therein); for example, by \cite[Th. 1.1]{Hag1}, one has the following (here we consider the 1-D case and choose the parameter there denoted by $\alpha$ equal to $1/2$): \par 
let $m$ be the mass of a quantum particle, $\hbar $ the reduced Planck constant and $V$ a potential such that $V\in C^{2}(\RE)$, $- c_{1} \le V(x)\le c_{2}\,e^{Mx^{2}}$ for some positive constants $c_{1},c_{2},M$, with $V''$ uniformly Lipschitz on compact subsets and let $(\sigmaq,\sigmap,q,p)\in\CO^{2}\times\RE^{2}$, with $\sigmaq,\sigmap \neq 0$ such that $\Re(\sigmap\,\sigmaq^{-1}) = |\sigmaq|^{-2} > 0$; define
$$
A_{t}:=\int_{0}^{t}ds\left(\frac1{2m}\,p_s^{2}-V(q_s)\right)\,,\qquad \H\psi:=-\frac{\hbar^{2}}{2m}\,\psi''+V\psi\,,
$$
and
$$
\psi^{\hbar}(x)\equiv\psi^{\hbar}(\sigmaq_t,\sigmap_t,q_t,p_t;x) :=\frac{1}{(2\pi\hbar)^{1/4}\sqrt{\sigmaq_t}}\ \exp\left({-\,\frac{\sigmap_t}{4\hbar \sigmaq_t}\,(x-q_t)^{2}+\frac{i}{\hbar}\,p_t(x-q_t)}\right)\,, 
$$
where $\sigmaq_t,\sigmap_t,q_t$, and $p_t$ solve the Cauchy problem
$$
\begin{cases}
\dot \sigmaq_t=\frac{i}{2m}\,\sigmap_t\,\\
\dot \sigmap_t=2iV''(q_t)\sigmaq_t\,,\\
\dot p_t=-V'(q_t)\,,\\
\dot q_t=\frac1{m}\,p_t\,,\\
\sigmaq_0=\sigmaq\,,\ \sigmap_0=\sigmap\,,\ q_0=q\,,\ p_0=p\,,
\end{cases}
$$
($\,\dot{\,}$ and $'$ denotes time and space derivatives respectively). Then for any $T>0$ and each positive $\lambda<1/2$, there exist positive constants $C_{\circ}$ and $h_{\circ}$ such that $\hbar<h_{\circ}$ implies
\be\label{Hage}
\big\|e^{-i\frac{t}{\hbar}\H}\,\psi^{\hbar}(\sigmaq,\sigmap,q,p)-e^{\frac{i}{\hbar}A_{t}}\,\psi^{\hbar}(\sigmaq_t,\sigmap_t,q_t,p_t)\big\|_{L^{2}(\RE)}\le C_{\circ}\hbar^{1/2 - \lambda}\,,\qquad |t|\le T\,.
\ee
Notice that, for the free Hamiltonian $\H_{0}$ (i.e., $V=0$), by Fourier transform one gets the exact correspondence 
\be\label{free}
e^{-i\frac{t}{\hbar}\H_{0}}\,\psi^{\hbar}(\sigmaq,\sigmap,q,p)=e^{\frac{i}{\hbar}\Szero_{t}}\,\psi^{\hbar}(\sigmaq_t,\sigmap,q_t,p)\,,
\ee
with
\begin{equation}\label{freeAsigmaq}
\Szero_{t}=\frac{p^{2}t}{2m}\,,\quad 
\sigmaq_t = \sigmaq+\frac{i\sigmap t}{2m}\,,\quad q_t = q+\frac{pt}{m}\,,
\end{equation}
i.e., \eqref{Hage} holds with $C_{\circ}=0$.\par
In the present paper we  provide  a result similar to the one stated in Eq. \eqref{Hage} whenever $V$ is a distributional potential describing a delta-interaction. Hence, here we consider the case $\H=\H_{\alpha}$, where 
$\H_{\alpha}$, $\alpha\in\RE\cup\{\infty\}$, is defined by means of the bounded from below, closed quadratic form 
$$
\dom(Q_{\alpha})=H^{1}(\RE)\,,\qquad
Q_{\alpha}(\psi):=\|\psi'\|^{2}_{L^{2}(\RE)}+\alpha\,|\psi(0)|^{2}\,,
$$
$$
\dom(Q_{\infty})=H_{0}^{1}(\RE_{-})\oplus\H^{1}_{0}(\RE_{+})\,,\qquad
Q_{\infty}(\psi):=\|\psi'\|^{2}_{L^{2}(\RE)}
$$
(here $H^{1}(\RE):=\{\psi\!\in\! L^{2}(\RE)\,|\,\psi'\!\in\! L^{2}(\RE)\}$ and $H_{0}^{1}(\RE_{\pm}):=\{\psi\!\in\! L^{2}(\RE_{\pm})\,|\,\psi'\!\in\! L^{2}(\RE_{\pm})\,,\ \psi(0_{\pm})\!=\!0\}$ are the usual Sobolev spaces of order one). 
Such a family of self-adjoint operators describes a set of self-adjoint extensions of the symmetric operator $\H^{\circ}_{0}$ given by the restriction of the free Hamiltonian $\H_{0}$ (corresponding to the choice $\alpha=0$) to the set $\C^{\infty}_{c}(\RE\backslash\{0\})$ (see, e.g.,  \cite[Ch. I.3]{AGHH}, \cite[Ch. 8]{Teta} and Section \ref{ss:Halpha} below for more details).

Evidently, since our potential lacks any regularity, we cannot mimic the proofs provided in \cite{Hag1} and related papers; however we take inspiration from relations \eqref{Hage} and \eqref{free}.
Some considerations about the semiclassical limit with a delta potential, based on the $\hbar\to 0$ limit of the continuation to imaginary time of the heat kernel of $H_{\alpha}$, have appeared in \cite[Sec. 3.4]{GS}; our approach here is quite different. 

To simplify the exposition, we restrict the attention to coherent states $\psi^{\hbar}(\sigma,\breve\sigma,q,p)$ with $q p \neq 0$, thus excluding cases with either $q=0$ (which is the support of the distributional potential) or $p=0$ (which gives, as regards the classical dynamics, no evolution).
\par
From now on we fix $\sigmazero>0$, $\sigmap=\sigmazero^{-1}$ and define, for any $\sigma\in\CO$, $\xi\equiv(q,p)\in\RE^{2}$, 
\begin{equation}\label{initial}
\psi^{\hbar}_{\sigma,\xi}:\RE\to\CO\,,\quad
\psi^{\hbar}_{\sigma,\xi}(x):=\frac{1}{(2\pi\hbar)^{1/4}\sqrt{\sigma}}\ \exp\left({-\frac{1}{4\hbar \sigmazero\sigma}\,(x-q)^{2}+\frac{i}{\hbar}\,p(x-q)}\right)
\end{equation}
and, for any $\sigma\in\CO$, $x\in\RE$, 
\begin{equation}\label{phidef}
\phi^{\hbar}_{\sigma,x}:\RE^{2}\to\CO\,,\quad
\phi^{\hbar}_{\sigma,x}(\xi):=\psi^{\hbar}_{\sigma,\xi}(x)\,.
\end{equation}
Using such notations, Eq. \eqref{free} can be re-written as
\begin{equation}\label{free2}
\big(e^{-i\frac{t}{\hbar}\H_{0}}\,\psi^{\hbar}_{\sigmazero,\xi}\big)(x)=e^{\frac{i}{\hbar}\Szero_t}\,\big(e^{itL_{0}}
\phi^{\hbar}_{\sigma_{t},x}\big)(\xi)\,, 
\end{equation}
where 
$$
\sigma_{t}:=\sigmazero+\frac{it}{2m\sigmazero}
$$
and $e^{itL_{0}}$ is the realization in $L^{\infty}(\RE^2)$ of the strongly continuous (in $L^{2}(\RE^{2})$) group of evolution generated by the self-adjoint operator 
$$
L_{0}:=-\,i\,X_{0} \cdot \nabla\,,\quad X_{0}(q,p):=\left(\frac{p}{m},0\right)\,,
$$ 
i.e.,  $e^{itL_{0}}f(q,p) = f(q+\frac{p}{m}\,t,p)$. Now a suggestion is clear:
since $\H_{\alpha}$ is a self-adjoint extension of $\H^{\circ}_{0}=\H_{0}\!\!\upharpoonright\!\C^{\infty}_{c}(\RE\backslash\{0\})$, one could try to approximate $e^{-i\frac{t}{\hbar}\H_{\alpha}}\,\psi^{\hbar}_{\sigmazero,\xi}$ by replacing $L_{0}$ with $\tilde L$, a self-adjoint extension   of $L^{\!\circ}_{0}:=L_{0}\!\!\upharpoonright\! \C^{\infty}_{c}({\mathscr M}_{0})$, ${\mathscr M}_{0}:=\RE^{2} \,\backslash\, \{(0,p) \,|\, p\!\in\!\RE\}$, and transforming $\phi^{\hbar}_{\sigma_{t},x}$ using the realization in $L^{\infty}(\RE^2)$, if any, of $e^{it\tilde L}$. \par 
In the following at first we provide (see Section \ref{s:2}) a family $L_{\beta}$, $\beta\!\in\!\RE\cup\{\infty\}$, of self-adjoint extensions of $L^{\!\circ}_{0}$ such that, for any $f\in L^{2}(\RE^{2})$, 
\begin{equation}\label{Lbeta-group}
\big(e^{itL_{\beta}}f\big)(q,p)=\big(e^{itL_{0}}f\big)(q,p)-\frac{\HE (-tqp)\,\HE \!\left(\frac{|pt|}{m}-|q|\right)}{1-\sgn(t)\,\frac{2i|p|}{m\beta}}\,\big(e^{itL_{0}}f_{ev}\big)(q,p)\,;
\end{equation}
here $\HE$ denotes the Heaviside function (namely, $\HE(\xi) = 1$ for $\xi > 0$ and $\HE(\xi) = 0$ for $\xi < 0$) and $f_{ev}(\xi):=f(\xi)+f(-\xi)$ (see Proposition \ref{exp-cl}). This also shows that $e^{itL_{\beta}}$ is a group of evolution in $L^{\infty}(\RE^2)$. Notice that the case $\beta=\infty$ corresponds to complete reflection due to the infinite barrier at the origin, while $\beta\in\RE\backslash\{0\}$ allows transmission ($\beta=0$ gives the free generator $L_{0}$). Therefore the case $\beta\in\RE\backslash\{0\}$ introduces ``extra'' classical paths going beyond the singularity; this resembles, in some way, the geometric theory of diffraction by Keller (see \cite{K1} and \cite{K2}). \par  In Section \ref{ss:3.1} we prove the following 
\begin{theorem}\label{t:1} Let $\beta=2\alpha/\hbar$. Then, there exists a constant $C>0$ such that, for any $\lambda_1,\lambda_2>0$, for any $t\in\RE$ and for any $\xi\equiv(q,p)\in\RE^{2}$ with $qp\not=0$, there holds
\begin{equation}\label{t1}
\begin{aligned}
&\left\|e^{-i\frac{t}{\hbar}\H_{\alpha}}\,\psi^{\hbar}_{\sigmazero,\xi}-
e^{\frac{i}{\hbar}\Szero_{t}} \big(e^{itL_{\beta}} \phi^{\hbar}_{\sigma_{t},(\cdot)}\big)(\xi)\right\|_{L^{2}(\RE)}\\
& \le C \Bigg[\!\left(\! {\hbar \over (m|\alpha|\sigma_0)^{2/3} }\!\right)^{\!\!\frac32-\ve_1}\!\! +  \,e^{-{1 \over 2} \Big(\! {(m|\alpha|\sigma_0)^{2/3} \over \hbar }\!\Big)^{\!\!2\ve_1}} \!\!
 + e^{-{\sigma_0 p^2 \over \hbar }}\! \bigg(\!\left({\hbar \over (m|\alpha|\sigma_0)^{2/3}}\!\right)^{\!\!{3 \over 2}-{\frac32 \lambda_2 }} \!+ e^{-{1 \over 2}\left(\!{(m|\alpha|\sigma_0)^{2/3} \over \hbar }\right)^{\!2\lambda_2}} \bigg)  \\
& \qquad\; +  e^{- {q^2 \over 8\hbar\sigma_0^2}} + e^{- {m |\alpha|\,|q| \over 8 \hbar^2}} + e^{- \frac{(q+pt/m)^2}{4\hbar|\sigma_{t}|^2}}\Bigg]\,. 
\end{aligned}
\end{equation}
\end{theorem}
Thus, whenever $t$ is not too close to the collision time $t_{coll}(\xi):=-\frac{mq}{p}$, $\xi\equiv(q,p)$  (look at  the last term in the above estimate), our approximation is better than the one given in \eqref{Hage} for regular potentials. In different terms, one has the following result (see Subsection \ref{proof-coroll}):
\begin{corollary}\label{c:1}
Let $\beta=2\alpha/\hbar$. Then, for any $0<\lambda<3/2 $ there exist constants $h_*,C_*,c_{0}>0$ such that 
$$
\underline h:= \max\left\{{\hbar\, \sigmazero^2 \over q^2 }\,,\, {\hbar \over (m |\alpha| \sigma_0)^{2/3} }\right\} < h_*
$$
implies 
$$\left\|e^{-i\frac{t}{\hbar}\H_{\alpha}}\,\psi^{\hbar}_{\sigmazero,\xi}-
e^{\frac{i}{\hbar}\Szero_{t}} \big(e^{itL_{\beta}} \phi^{\hbar}_{\sigma_{t},(\cdot)}\big)(\xi)\right\|_{L^{2}(\RE)}\le C_{*}\,\underline h^{\frac32- \lambda}
$$ 
for any $t\in\RE$, $\xi\equiv(q,p)\in\RE^{2}$ with $qp\not=0$, such that
\begin{equation}\label{ttcoll}
\big|t-t_{coll}(\xi)\big|\ge c_{0}\,|t_{coll}(\xi)|\,\sqrt{\Big(\frac32-\lambda\Big)\,\underline h\,|\!\ln \underline h|}\;.
\end{equation}
\end{corollary}
Notice that, if   one approximates the quantum dynamics  with the classical dynamics corresponding 
to the infinite barrier case (i.e., using the operator $L_{\infty}$), which would seem the most natural choice, the estimates cannot be better (whenever $\alpha\not=\infty$) than $O\big({\hbar |p| \over m |\alpha|}\big)$ (see Remark \ref{r:dirichlet}):
$$\left\|e^{-i\frac{t}{\hbar}\H_{\alpha}}\,\psi^{\hbar}_{\sigmazero,\xi}-
e^{\frac{i}{\hbar}\Szero_{t}} \big(e^{itL_{\infty}} \phi^{\hbar}_{\sigma_{t},(\cdot)}\big)(\xi)\right\|_{L^{2}(\RE)}\ge C_{*}\,{\hbar |p| \over m |\alpha|}\;.
$$ 
   
Moreover, the constraint  $t\not=t_{coll}$ does not affect the semi-classical approximation for large times. Indeed, see Theorem \ref{t:2} below, we can handle the approximations of the wave operators: denoting with $\Omega^{\pm}_{\alpha}$ the wave  operators defined, as usual, by the limits in 
$L^{2}(\RE^{2})$
$$
\Omega^{\pm}_{\alpha}f:=\lim_{t\to\pm\infty}e^{i\frac{t}{\hbar}\H_{\alpha}}e^{-i\frac{t}{\hbar}\H_{0}}f
$$
and by $W^{\pm}_{\beta}$ the corresponding classical objects (see Subsection \ref{waveop})
\begin{equation*}
W^{\pm}_{\beta}f:=\lim_{t\to\pm\infty}e^{i tL_{0}}e^{ - i  t L_{\beta}}f
\end{equation*}
(here the limits hold both pointwise in $\RE^{2}$ and, if $f = \psih_{\sigma,\xi}$ is a coherent state of the form \eqref{initial}, in $L^{2}(\RE,dx)$), one has the following (see Subsection \ref{proof2} for the proof) 
\begin{theorem}\label{t:2} Let $\beta=2\alpha/\hbar$. Then, there exists a constant $C>0$ such that for any $\lambda>0$ and for any $\xi\equiv(q,p)\in\RE^{2}$ with $qp\not=0$, 
\begin{equation}\label{t2_1}
\left\|\Omega^{\pm}_{\alpha}\psi^{\hbar}_{\sigmazero,\xi}- \big(W^{\pm}_{\beta}\phi^{\hbar}_{\sigmazero,(\cdot)}\big)(\xi)\right\|_{\L^{2}(\RE)}\le C\left(\!
\left(\! {\hbar \over (m|\alpha|\sigma_0)^{2/3} }\!\right)^{\!\!\frac32-\ve}\!  + \,e^{-{1 \over 2} \Big(\! {(m|\alpha|\sigma_0)^{2/3} \over \hbar }\!\Big)^{\!\!2\ve}}\!+ e^{-  {\sigmazero^2\, p^2 \over \hbar}} + e^{- \frac{q^2}{4\hbar\sigmazero^2}}\right).
\end{equation}
Analogously, for the corresponding scattering operators $S_{\alpha}:=(\Omega_{\alpha}^{+})^{*}\Omega^{-}_{\alpha}$ and $S^{cl}_{\beta}:=(W_{\beta}^{+})^{*}W^{-}_{\beta} $ one has
\begin{equation}\label{t2_2}
\begin{aligned}
&\left\|S_{\alpha}\psi^{\hbar}_{\sigmazero,\xi}- \big(S^{cl}_{\beta}\phi^{\hbar}_{\sigmazero,(\cdot)}\big)(\xi)\right\|_{\L^{2}(\RE)} \\ 
& \le C\left(\!\left(\! {\hbar \over (m|\alpha|\sigma_0)^{2/3} }\!\right)^{\!\!\frac32-\ve}\!  + \,e^{-{1 \over 2} \Big(\! {(m|\alpha|\sigma_0)^{2/3} \over \hbar }\!\Big)^{\!\!2\ve}}\!+ e^{-  {\sigmazero^2\, p^2 \over \hbar}} +  e^{- {q^2 \over 8\hbar\sigma_0^2}} + e^{- {m |\alpha|\,|q| \over 8 \hbar^2}}\right). 
\end{aligned}
\end{equation}
\end{theorem}

\begin{corollary}\label{c:2}
Let $\beta=2\alpha/\hbar$. Then, for any $0<\lambda<3/2 $ there exist  constants  $h_*,C_*>0$ such that 
$$
\underline{\underline h}:= \max\left\{{\hbar\, \sigmazero^2 \over q^2 }, {\hbar \over (m |\alpha| \sigma_0)^{2/3} }, { \hbar \over \sigmazero^2\, p^2 }\right\} < h_*
$$
implies 
\[ 
\left\|\Omega^{\pm}_{\alpha}\psi^{\hbar}_{\sigmazero,\xi}- \big(W^{\pm}_{\beta}\phi^{\hbar}_{\sigmazero,(\cdot)}\big)(\xi)\right\|_{\L^{2}(\RE)}\le C_*\, \underline{\underline h}^{\frac32-\lambda}\;,
\]
\[
\left\|S_{\alpha}\psi^{\hbar}_{\sigmazero,\xi}- \big(S^{cl}_{\beta}\phi^{\hbar}_{\sigmazero,(\cdot)}\big)(\xi)\right\|_{\L^{2}(\RE)}  
\le  C_*\, \underline{\underline h}^{\frac32-\lambda}\;.
\]
\end{corollary}


\section{Singular perturbations of the free classical dynamics\label{s:2}}

Let $h_{0}(q,p)= p^2/(2m)$ be the Hamiltonian function of a classical free particle in $\RE$ and let $X_{0}(q,p)=(p/m,0)$ be the related Hamiltonian vector field. In this connection, consider the differential operator 
$$
L:\S'(\RE^{2})\to \S'(\RE^{2})\,,\qquad Lf:= - \,i\,X_{0}\cdot\nabla f
$$
in the space of tempered distributions $\S'(\RE^{2})$, namely the dual of the Schwartz space $\S(\RE^{2})$.
Let ${\mathscr M}_{0}$ denote the manifold ${\mathscr M}_{0}:=\RE^{2}\,\backslash\,\{(0,p) \,|\, p\!\in\!\RE\}$; since the flow of $X_{0} \!\upharpoonright\! {\mathscr M}_{0}$ is clearly not complete, by Povzner's theorem (see \cite{Povz} and \cite[Th. 2.6.15]{AM}), $L\!\upharpoonright\!\C^{\infty}_{c}({\mathscr M}_{0})$ is not essentially self-adjoint in $L^{2}({\mathscr M}_{0},dq\,dp)\equiv L^{2}(\RE^{2},dq\,dp) \equiv L^{2}(\RE^{2})$. In the sequel we proceed to characterize (some of) its self-adjoint extensions.

We denote by 
\begin{gather*}
L_{0}:\dom(L_{0})\subseteq L^{2}(\RE^{2})\to L^{2}(\RE^{2})\,,\qquad (L_{0}f)(q,p)  = -\,i\,{p \over m}\,\frac{\partial f }{\partial q}(q,p)\,, \\
\dom( L_{0}) := \big\{ f \in L^2(\RE^2) \;\big|\; Lf\in L^2(\RE^2) \big\} \,, 
\end{gather*}
the maximal realization of $L_{0}$ in $L^{2}(\RE^{2})$. By means of the partial Fourier transform
$$
\tilde f(k,p):=\frac1{\sqrt {2\pi}}\int_{\RE}\! dq\; e^{-ikq}\,f(q,p)\,,
$$
$L_{0}$ is unitarily equivalent to the multiplication operator 
\begin{gather*}
\tilde{L}_{0}:\dom(\tilde{L}_{0})\subseteq L^{2}(\RE^{2})\to L^{2}(\RE^{2})\,,\qquad (\tilde{L}_{0} \tilde f)(k,p)  = {k\,p \over m}\, \tilde f(k,p)\,, \\
\dom(\tilde{L}_{0}):= \big\{\tilde f\in L^{2}(\RE^{2}) \;\big|\; (k,p)\mapsto k\, p\, \tilde f(k,p) \in L^{2}(\RE^{2})\big\}\,;
\end{gather*}
thus, $L_{0}$ is self-adjoint and its spectrum is purely absolutely continuous, $\sigma(L_{0})=\sigma_{ac}(L_{0})=\RE$ (see \cite[Ch. X, Ex. 1.19]{Kato}).

Let us now define the linear map
\begin{equation}\label{tau}
\gamma : \S(\RE^2) \to \S(\RE)\,, \quad (\gamma f)(p) := {1 \over \sqrt{2\pi}}\int_{\RE}\! dk\; \tilde f(k,p)\,.
\end{equation}

\begin{lemma} The map $\gamma$ extends to a bounded operator $\gamma :\dom(L_{0}) \to  L^2(\RE, |p|\,dp)$, where $\dom(L_{0}) \subset L^2(\RE)$ is endowed with the graph norm. Moreover, $\ker(\gamma)$ is dense in $L^{2}(\RE^{2})$ and $(\gamma f)(p) = f(0,p)$ for any $f \in \S(\RE^{2})$.
\end{lemma}
\begin{proof} At first let us notice that, by the definition of $\dom(L_{0})$, we have $f \!\in\! \dom(L_{0})$ if and only if  $\tilde f \!\in\! L^2(\RE^2, (1+|k\,p|^2)dk\,dp)$. Hence $\S(\RE^{2})$ is a dense subset of $\dom(L_{0})$ (w.r.t. the graph norm). By inverse Fourier transform, for any $f\in \S(\RE^{2})$ we have
$$
f(q,p)=\frac1{\sqrt{2\pi}}\int_{\RE}\! dk\; e^{ikq}\, \tilde f (k,p)
$$ 
and so $f(0,p)= (\gamma f)(p)$ for any $f \in \S(\RE^{2})$. Therefore, since  
$$
K:= \big\{f\in \S(\RE^{2})\;\big|\; f(0,p)=0\big\}\subseteq\ker(\gamma)\,,
$$
and $\overline K=L^{2}(\RE^{2})$, we get $\overline{\ker(\gamma)}=L^{2}(\RE^{2})$. 

Finally, for any $f\in \S(\RE^{2})$ there holds the following chain of inequalities:
\begin{align*}
\|\gamma f\|^{2}_{L^2(\RE, |p|\,dp)}
= &\;\int_{\RE}\! dp\; |p|\;  |\gamma f(p)|^2 
= {1 \over 2\pi} \int_{\RE}\!  dp\; |p|\, \left|  \int_{\RE}\!  dk\;  \tilde f (k,p)\right|^2 \\ 
\leq &\; \frac1{2\pi}\int_{\RE}\! dp\; |p|\, \left( \int_{\RE}\! dk\; (1+|kp|^2)\, |\tilde f (k,p)|^2 \right)\! \left( \int_{\RE} \frac{dk'}{1+|k'p|^2} \right) \\
= &\; \frac1{2\pi}\left( \int_{\RE}\! \frac{dy }{1+y^2} \right) \int_{\RE}\! dp\, \int_{\RE}\!  dk\; (1+|kp|^2)\, | \tilde f (k,p)|^2\\
= &\; {1 \over 2} \left(\|f\|^{2}_{L^2(\RE^{2})}+\|L_{0}f\|^{2}_{L^2(\RE^{2})}\right) .
\end{align*}
This proves the boundedness of $\gamma :\dom(L_{0}) \to  L^2(\RE, |p|\,dp)$ with respect to the graph norm in $\dom(L_{0})$, thus concluding the proof.
\end{proof}
Since $L_{0}$ is invariant with respect to translations in the $q$-variable, posing $R^{0}_z := (L_{0} -z)^{-1}$ for $z\in\CO\backslash\RE$ one gets
\begin{equation}\label{res}
(R^{0}_z f)(q,p) = \int_{\RE}\!  dq'\, g_z(q - q',p)\, f(q',p)
\end{equation}
for some kernel $g_{z}(q,p)$. By partial Fourier transform, one has
\begin{equation}\label{tg}
\tilde g_{z}(k,p) = {1 \over \sqrt{2\pi}}\; \frac{1}{{kp \over m} - z}
\end{equation}
and so
$$
g_z(q,p) = \frac1{2\pi}\int_{\RE}\!dk\;  \frac{e^{ikq}}{{kp \over m} - z}\,.
$$
Computing the above integral one obtains
$$
g_z(q,p) = \left\{\!\begin{array}{lr}
\dd{0} 									& \dd{q\, p < 0}		\\ 
\dd{\frac{i\,m}{p}\,e^{i m z q/p}}	& \dd{q > 0,\, p > 0}	\\
\dd{-\,\frac{i\,m}{p}\,e^{i m z q/p}}		& \dd{q < 0,\, p < 0}
\end{array}\right. \qquad \mbox{for $\Im z > 0$}\,,
$$
$$
g_z(q,p) = \left\{\!\begin{array}{lr}
\dd{0} 									& \dd{q\, p > 0}		\\ 
\dd{\frac{i\,m}{p}\,e^{i m z q/p}}	& \dd{q > 0,\, p < 0}	\\
\dd{-\,\frac{i\,m}{p}\,e^{i m z q/p}}		& \dd{q < 0,\, p > 0}
\end{array}\right. \qquad \mbox{for $\Im z < 0$}\,,
$$
i.e., more compactly,
$$
g_z(q,p) = \HE(q\,p \Im z)\;\sgn(\Im z)\;\frac{i\,m}{|p|}\,{e^{i m z q/p}}
$$
(recall that  $\HE$ indicates the Heaviside step function).

For any  $z\in\CO\backslash\RE$, we define the bounded linear map
$$
G_{z}:L^{2}(\RE,|p|^{-1}dp)\to L^{2}(\RE^{2})\,,\qquad G_{z}:=(\gamma\, R^0_{\bar z})^{*}\,.
$$
Noting the identity $\overline{g_{\bar z}(-q,p)} = g_z(q,p) $, by Eq. \eqref{res} we get
$$
(G_{z}u)(q,p) = g_{z}(q,p)\,u(p)\,.
$$
Hence, by the first resolvent identity we infer (see \cite[Lem. 2.1]{P01}, paying attention to the different sign convention for the resolvent used therein)
$$
(z-w)G^{*}_{\bar w}G_{z} = \gamma(G_{z}-G_{w}) \quad \mbox{for any $z,w \in \CO \backslash \RE$, $z \neq w$}\,.
$$
By \eqref{tau} and \eqref{tg} we have
\begin{align*}
\big(\gamma (g_{z}-g_{w})\big)(p) 
=\; & \frac{1}{2\pi}\int_{\RE}\! dk\left(\frac1{{kp \over m} - z}-\frac1{{kp \over m} - w}\right)\\
=\; & \frac{m}{2\pi\,|p|}\int_{\RE}\! dh\left({1 \over \sgn(p)\,h - z} - {1 \over \sgn(p)\,h - w}\right) \\
=\; & \frac{m}{2\pi\,|p|}\; (z - w)\int_{\RE}\! dh\; {1 \over (h - \sgn(p)\,z)(h - \sgn(p)\,w)} \\
=\; & \frac{i\,m}{2\,|p|}\,\big(\sgn(\Im z)-\sgn(\Im w)\big) \;.
\end{align*}
Therefore, for any given $\ac \in (\RE \backslash \{0\}) \cup \{\infty\}$ and for all $z \in \CO \backslash \RE$, the linear map
$$
M^{\ac}_{z} : \dom(M^{\ac}_{z})\subset L^{2}(\RE,|p|^{-1}dp) \to L^{2}(\RE,|p|dp)\,, \quad
(M^{\ac}_{z}u)(p) := m^{\ac}_{z}(p)\,u(p)\,, 
$$
$$ 
\dom(M^{\ac}_{z}):=L^{2}(\RE,|p|^{-1}dp) \cap L^{2}(\RE,|p|dp)\,,\quad m^{\ac}_{z}(p):= \frac1\ac - \sgn(\Im z)\,\frac{i\,m}{2\,|p|}
$$
satisfies the identities (see \cite[Eq.s (5) and (7)]{P01}, recalling again that we are using a different sign convention for the resolvent)
$$
(M_{z}^{\ac})^{*}=M_{\bar z}^{\ac}\,,\qquad  M_{z}^{\ac}-M_{w}^{\ac}=(w-z)G^{*}_{\bar w}G_{z}
$$
(here $L^{2}(\RE,|p|^{-1}dp)$ and $L^{2}(\RE,|p|dp)$ are considered as a dual couple with respect to the duality induced by the scalar product in $L^{2}(\RE)$). Incidentally, let us remark that the above map $M^{\ac}_{z}$ can be equivalently characterized as
\begin{equation}\label{MtauG}
M^{\ac}_{z} = {1 \over \ac} - \widehat{\gamma}\, G_{z}\,,
\end{equation}
where $\widehat{\gamma}$ is the extension of the trace map \eqref{tau} defined as
$$
(\widehat{\gamma}\, f)(p) := {1 \over 2}\, \big(f(0_{+},p)\, + f(0_{-},p) \big) \,.
$$
Since $m^{\ac}_{z}(p)\not=0$ for any $p\in\RE$, 
\begin{equation}\label{LaM}
\Lambda^{\ac}_{z}:=(M_{z}^{\ac})^{-1}:L^{2}(\RE,|p|dp)\to L^{2}(\RE,|p|^{-1}dp) \,,\quad
(\Lambda^{\ac}_{z}u)(p):=\frac{u(p)}{m^{\ac}_{z}(p)}
\end{equation}
is a well-defined, bounded linear map for any $z\in\CO\backslash\RE$.
\par
In addition, let us consider the projector on even functions (here either $\rho(p)=|p|$ or $\rho(p)=|p|^{-1}$)
\be\label{proj}
\Pi : L^2(\RE,\rho dp) \to L^2(\RE,\rho dp)\,, \qquad (\Pi f)(p) := {1 \over 2}\,\big( f(p) + f(-p) \big)
\ee
and notice that, since $(\Pi (f \!\cdot\! g))(p) = f(p)\,(\Pi g)(p)$ for any even $f$, we have in particular
$$ 
\Pi\, M^{\ac}_{z} - M^{\ac}_{z}\, \Pi = 0\,, \qquad \Pi \,\Lambda^{\ac}_{z} - \Lambda^{\ac}_{z}\, \Pi = 0\,.
$$
Then, by \cite[Th. 2.1]{P01} here employed with $\tau:=\Pi\gamma$, we obtain the following
\begin{theorem}\label{t:resolvent}
For any $\ac\in (\RE \backslash \{0\}) \cup \{\infty\}$, the linear bounded operator
$$
R^{\ac}_{z} := R_{z}^{0} + G_{z}\, \Pi\, \Lambda^{\ac}_{z}\, \Pi\, G^{*}_{\bar z}\quad \mbox{with $z\in\CO\backslash\RE$}\,,
$$ 
is the resolvent of a self-adjoint extension $L_{\ac}$ of the densely defined, closed symmetric operator $L_{0}\!\upharpoonright\!\ker(\gamma)$. Such an extension acts on its domain
$$
\dom(L_{\ac}):= \big\{f\in L^{2}(\RE^{2}) \;\big|\; f=f_{z}+G_{z}\,\Lambda^{\ac}_{z}\,\Pi\,\gamma f_{z}\,,\; f_{z}\in\dom(L_{0})\big\}
$$
by
\begin{equation}\label{Lac}
(L_{\ac}-z)f=(L_{0}-z)f_{z}\,.
\end{equation}
\end{theorem}

\begin{remark}\label{remBC}
Let us consider the function
$$ 
\phi := \Lambda^{\ac}_{z}\,\Pi\,\gamma\, f_z \in L^2(\RE,|p|^{-1}dp)\,,
$$
and notice that by construction we have $\phi \in \ran (\Pi)$, i.e., $\Pi \phi = \phi$. Then, recalling the previous relations \eqref{MtauG} - \eqref{LaM}, by a standard computation (see, e.g., \cite{P01}) it can be  inferred that any $f \in \dom(L_{\beta})$ fulfils the boundary condition
\begin{equation}\label{BCphi}
\Pi\,\widehat{\gamma}\,f = {1 \over \ac}\,\phi\,.
\end{equation}
Moreover, on account of the basic identity $(L - z) G_z \phi = \phi\,\delta_{\Sigma_{0}}$, 
where $\phi\delta_{\Sigma_{0}}$ is the tempered distribution supported on $\Sigma_{0}:=\{(q,p)\!\in\!\RE^{2}\,|\, q=0\}$ defined by 
$$
\big(\phi\,\delta_{\Sigma_{0}}\big)(\varphi):=\int_{\Sigma_{0}}\!\!dp\;\phi(p)\,\varphi(0,p)\,,\quad\; \mbox{for any $\varphi\in\S(\RE^{2})$}\,,
$$
from the above relations \eqref{Lac} - \eqref{BCphi} we readily infer that
\begin{equation*}
L_{\ac} f = L f - \phi\,\delta_{\Sigma_{0}} \,,
\end{equation*}
or, equivalently,
\[
L_{\ac} f = L f - \ac\, \big(\Pi\,\widehat{\gamma}\, f\big)\, \delta_{\Sigma_{0}}\,.
\]
\end{remark}

\subsection{The classical perturbed dynamics}

The action of the unitary group $e^{- i t L_{\ac}}$ ($t \in \RE$) describing the dynamics induced by $L_{\ac}$ can be explicitly characterized by means of elementary functional calculus starting from the corresponding resolvent operator $R^{\ac}_{z}$ introduced in Theorem \ref{t:resolvent}.

In the sequel, we indicate with $e^{- i t L_{0}}$ ($t \in \RE$) the free unitary group
\begin{equation}\label{ExpFree}
\big(e^{- i t L_{0}} f\big)(q,p) = f\Big(q - {p t \over m} , p\Big)\,.
\end{equation}

\begin{proposition}\label{exp-cl} Let $\ac \in (\RE\backslash\{0\}) \cup \{\infty\}$ and $f \in L^2(\RE^2)$. Then, for all $t \in \RE$ there holds
\begin{equation}\label{ftal1}
\big(e^{- i t L_{\ac}} f\big)(q,p) = \big(e^{- i t L_{0}} f\big)(q,p) \,-\, \frac{\HE(t\,q\,p)\; \HE\big({|p\,t| \over m} - |q|\big)}{1 + \sgn(t)\,{2i\,|p| \over m\,\ac}}\, \big(e^{- i t L_{0}} f_{ev}\big)(q,p)\;,
\end{equation}
where $f_{ev}(q,p):=f(q,p)+f(-q,-p)$.
\end{proposition}
\begin{proof} Firstly, let us remark that on account of the results mentioned in Theorem \ref{t:resolvent} we have
\begin{align*}
& \big(R^{\ac}_{z} f\big)(q,p) = \big(R^{0}_{z} f\big)(q,p) + \big(G_z \Lambda^{\ac}_z \Pi \gamma R^{0}_{z} f\big)(q,p) \\
& = \int_{\RE}\!\! dq' \Bigg[g_z(q\!-\!q',p)\,f(q',p) + {g_z(q,p) \over {1 \over \ac}\!-\! \sgn(\Im z)\,{i m \over 2|p|}}\, 
{1 \over 2}\Big(\gamma g_z(\cdot -\! q',p)f(q',p) + \gamma g_z(\cdot -\! q',-p)f(q',-p) \Big)\Bigg] \\
& = \sgn(\Im z)\,\frac{i\,m}{|p|} \int_{\RE}\!\! dq' \Bigg[\HE\big((q-q')\,p \Im z\big)\;{e^{i m z (q-q')/p}}\, f(q',p) \\
& \hspace{1.5cm} - {\HE(q p \Im z) \over 1 + \sgn(\Im z)\,{2 i|p| \over m \ac}} \Big(
\HE(-q' p \Im z)\, e^{i m z (q-q')/p}\, f(q',p) 
+ \HE(q' p \Im z)\,e^{i m z (q+q')/p}\, f(q',-p) \Big)\Bigg]\,.
\end{align*}
In the following we show how to derive Eq. \eqref{ftal1} for $t > 0$. Analogous arguments can be employed in the case $t < 0$ (see, in particular Eq. \eqref{2.11a} below), but for brevity we omit the details of the related computations.

By standard functional calculus one has
\begin{equation}\label{RLap}
\big(R^{\ac}_{z} f\big) = i \int_{0}^{\infty}\!\! dt\;e^{i z t}\,\big(e^{- i t L_{\ac}} f\big) \qquad \mbox{for\, $\Im z > 0$}\,.
\end{equation}
Inverting the Laplace transform in the above relation yields the following, for any $c > 0$, $t > 0$ and $f\in \dom(L_{\beta})$ (see \cite[Ch. III, Cor. 5.15]{EN}):
$$ 
\big(e^{-it L_{\ac}}f\big)= {1 \over 2\pi i} \lim_{n\to \infty}\int_{i c - n}^{i c + n}\! dz\; e^{-i z t}\,\big(R^{\ac}_{z} f\big) = {e^{c t} \over 2\pi i} \lim_{n\to \infty}\int_{-n}^{n}\! dk\; e^{- i t k}\,\big(R^{\ac}_{k + i c} f\big)\,. 
$$ 
Upon substitution of the previously mentioned expression for $R^{\ac}_{z} f$, keeping in mind that $\Im z = c > 0$ we obtain
\begin{align*}
& \big(e^{-it L_{\ac}}f\big)(q,p) \\
& = \frac{m}{|p|}\,{e^{c t} \over 2\pi} \lim_{n\to \infty}\int_{-n}^{n}\! dk\; e^{- i t k} \int_{\RE}\!\! dq' \Bigg[\HE\big((q-q')\,p\big)\;{e^{i m (k + i c) (q-q')/p}}\, f(q',p) \\
& \hspace{1.5cm} - {\HE(q p) \over 1 + {2 i|p| \over m \ac}}\, \Big(
\HE(-q' p)\, e^{i m (k + i c) (q-q')/p}\, f(q',p) 
+ \HE(q' p)\,e^{i m (k + i c) (q+q')/p}\, f(q',-p) \Big)\Bigg] \\
& = \frac{m}{|p|}\, {e^{c (t - m q/p)} \over 2\pi} \lim_{n\to \infty}\int_{-n}^{n}\! dk\; e^{i (m q/p - t) k} \int_{\RE}\!\! dq' \Bigg[ e^{- i m (k + i c) q'/p}\; \HE\big((q-q')\,p\big)\, f(q',p) \\
& \hspace{1.5cm} 
- {\HE(q p) \over 1 + {2 i|p| \over m \ac}}\, \Big(
e^{- i m (k + i c) q'/p}\; \HE(-q' p)\, f(q',p) 
+ e^{i m (k + i c) q'/p}\; \HE(q' p)\,f(q',-p) \Big)\Bigg]\,.
\end{align*}
The latter identity can be rephrased in terms of unitary Fourier transforms and of their inverses; more precisely, denoting these transforms respectively with $\mathfrak{F}$ and $\mathfrak{F}^{-1}$, we have
\begin{align*}
& \big(e^{-it L_{\ac}}f\big)(q,p) \\
& = \frac{m}{|p|}\, e^{c (t - m q/p)}\,
\mathfrak{F}^{-1} \Big(\mathfrak{F}\big(\HE\big((q-\,\cdot)\,p\big)\, f(\,\cdot\,,p)\big)\big((\ast + i c)m/p\big) \Big)(m q/p - t)
\\
& \hspace{0.5cm} 
- {\HE(q p) \over 1 +{2 i|p| \over m \ac}}\,\frac{m}{|p|}\, e^{c (t - m q/p)} \Bigg[
\mathfrak{F}^{-1} \Big(\mathfrak{F}\big(\HE\big(-(\cdot) p\big)\, f(\,\cdot\,,p)\big)\big((\ast + i c)m/p\big)\Big)(m q/p - t) \\
& \hspace{5.cm} +
\mathfrak{F}^{-1} \Big(\mathfrak{F}\big(\HE\big((\cdot)\,p\big)\, f(\,\cdot\,,-p)\big)\big(- (\ast + i c)m/p\big)\Big)(m q/p - t)\Bigg]\,.
\end{align*}

In view of the basic identity 
$$
\mathfrak{F}^{-1}\Big(\mathfrak{F}h\big(a((\ast) + ic)\big)\Big)(q) = {e^{c q} \over |a|}\,h(q/a) \qquad \mbox{for $a \in \RE \backslash \{0\}$, $c>0$, $q \in \RE$}\,,
$$
which holds true whenever ${e^{c \cdot} \over |a|}\,h(\cdot /a) \in L^2(\RE)$,  by elementary computations we obtain
\begin{align*}
& \big(e^{-it L_{\ac}}f\big)(q,p) \\
& = \HE\Big({p^2 t \over m}\Big)\,f\Big(q - {p t \over m},p\Big) - {\HE(q p)\, \HE\Big({p^2 t \over m} - q p\Big) \over 1 + {2 i|p| \over m \ac}}\, \Bigg[f\Big(q - {pt \over m},p\Big) + f\Big(\!-q + {p t \over m},-\,p\Big) \Bigg]\,.
\end{align*}
Taking into account  Eq. \eqref{ExpFree}, the above relation is equivalent to Eq. \eqref{ftal1} for $t > 0$, since in this case $\HE\big({p^2 t \over m}\big) = 1$ and $\HE(q p)\, \HE\big({p^2 t \over m} - q p\big) = \HE(q p) \HE\big({|p| t \over m} - |q|\big)$. Then Eq. \eqref{ftal1} is extended to hold for any $f\in L^{2}(\RE^{2})$ by density.

For $t < 0$, in place of Eq. \eqref{RLap} one should consider the identity $$\big(R^{\ac}_{z} f\big) = -i \int_{-\infty}^{0}\!\! dt\;e^{i z t}\,\big(e^{- i t L_{\ac}} f\big)\qquad \mbox{for\, $\Im z < 0$}\,. $$ Then, analogous computations  yield 
\begin{equation}\label{2.11a}
\begin{aligned}
& \big(e^{-it L_{\ac}}f\big)(q,p) \\
& = \HE\Big(\!-{p^2 t \over m}\Big)\,f\Big(q - {p t \over m},p\Big) - {\HE(- q p)\, \HE\Big(\!- {p^2 t \over m} + q p\Big) \over 1 - {2 i|p| \over m \ac}}\, \Bigg[f\Big(q - {pt \over m},p\Big) + f\Big(-q + {p t \over m},-p\Big) \Bigg]\,,
\end{aligned}
\end{equation}
which again reproduces the relation written in Eq. \eqref{ftal1}, thus proving the thesis.
\end{proof}
\begin{remark} Formula \eqref{ftal1} shows that $e^{-itL_{\beta}}$ defines a group of evolution in 
 $L^{\infty}(\RE^2)$.
\end{remark}
\begin{remark}\label{RemDir}
Notice that the free operator $e^{- i t L_{0}}$ maps real-valued functions into real-valued functions. The same does not hold true for the perturbed analogue $e^{- i t L_{\ac}}$, unless $\ac = \infty$ (corresponding to the case of Dirichlet boundary conditions, for which there is complete reflection); in this particular case, Eq. \eqref{ftal1} reduces to
\begin{align*}
&\big(e^{-it L_{\infty}}f\big)(q,p) \\
& = \left[1 - \HE(t\,q\,p)\;\HE\!\left({|p\,t| \over m} - |q|\right)\!\right] f\Big(q - {pt \over m},p\Big)
  -\,\HE(t\,q\,p)\; \HE\!\left({|p\,t| \over m} - |q|\right) f\Big(\!-q + {pt \over m},-\,p\Big)\,.
\end{align*}
This also justifies our introduction of the projector $\Pi$ defined in \eqref{proj}: it leads to a family of self-adjoint extensions containing the generator of the dynamics corresponding to complete reflection.
\end{remark}

\subsection{The classical wave operators and scattering operator}\label{waveop}
In this section we find explicit formulae for the classical wave operators defined by 
\be\label{w-class}
W^{\pm}_{\beta}f:=\lim_{t\to\pm\infty}e^{   i tL_{0}}e^{ - i  t L_{\beta}}f
\ee
and for the corresponding classical scattering operator 
\[
S^{cl}_\beta  := (W_\beta^{+})^{*}\, W_\beta^{-}\,.
\]
\begin{remark}
Notice that the classical wave operators for a couple of flows in the phase space $\RE^{2}$, $\varphi^{0}_{t}$ (the free one) and $\varphi_{t}$ (the interacting one),  are defined pointwise by 
$$
w^{\pm}:=\lim_{t\to\pm\infty}\varphi_{-t}\circ\varphi_{t}^{0}
$$
(see, e.g., \cite[Def. 3.4.4]{Th}). These induces the wave operators acting on functions defined by 
$$
W^{\pm}f:=\lim_{t\to\pm\infty}f(\varphi_{-t}\circ\varphi_{t}^{0})=\lim_{t\to\pm\infty}e^{   i tL_{0}}f(\varphi_{-t})\,.
$$
This justifies our definition, taking into account the evolution $e^{ - i  t L_{\beta}}$  which is not induced (unless $\beta=0$ and $\beta=\infty$) by a flow in the phase space.
\end{remark}
\begin{proposition}\label{prop-wave-cl} The limits in \eqref{w-class} exist pointwise for any $\xi\equiv(q,p)\!\in\! \RE^2$ with $qp\not=0$ and in $L^{2}(\RE^2)$ for any $f\in L^{2}(\RE^2)$; moreover, there holds
\be\label{WOcl+}
\big(W_\beta^{\pm} f\big)(q,p)  = f(q, p) - \; \frac{\HE(\mp q\,p)}{1 \pm {2i\,|p| \over m\,\ac}}\  f_{ev}(q, p)\,,  
\ee
where $f_{ev}:=f(q,p)+f(-q,-p)$. Furthermore, the scattering operator is given by 
\begin{equation}\label{Sbeta}
\big(S^{cl}_\beta f\big)(q,p) = f(q, p) - \, \frac{f_{ev}(q, p)}{1 - {2i\,|p| \over m\,\ac}}\;.
 \end{equation}
\end{proposition}
\begin{proof}
Using the results derived previously (see, in particular,  Eq. \eqref{ftal1},   and recall that $e^{ i t L_{0}}\,f(q,p) = f(q + pt/m,p)$ for all $t\in\RE$), it can be readily inferred that, for any $t \in \RE$,
\begin{align*}
& \big(e^{   i t L_{0}}\,e^{ -  i t L_{\ac}}\,f\big)(q,p) 
= f(q, p) - \HE\!\left(\!t\,\Big(q +  \frac{pt}{m}\Big) \,p\!\right)\, \HE\!\left(\!{|p t| \over m} - \Big|q + \frac{pt}{m}\Big|\!\right) \frac{f_{ev}(q, p)}{1 + \sgn(t)\,{2i\,|p| \over m\,\ac}}\,.
\end{align*}
Taking into account the limits 
\[
\lim_{t\to\pm\infty }  \HE\!\left(t\,\Big(q + \frac{pt}{m}\Big) \,p\right) = 1 \,,
\]
\[
\lim_{t\to+ \infty } \HE\!\left(\!{|p t| \over m} - \Big|q +  \frac{pt}{m}\Big|\!\right)  = \lim_{t\to+ \infty } \HE\!\left(\!{|p| t \over m} \left( 1- \Big|\frac{q m}{pt}+1\Big|\!\right)\right) =    \HE\!\left( - \sgn (p) q  \right)=    \HE\!\left( -  q\, p  \right)\,,
\]
and 
\[
\lim_{t\to-  \infty } \HE\!\left(\!{|p t| \over m} - \Big|q +  \frac{pt}{m}\Big|\!\right)  = \lim_{t\to- \infty } \HE\!\left(\! - { |p| t \over m} \left( 1- \Big|\frac{q m}{pt}+1\Big|\!\right)\right) =    \HE\!\left(  \sgn (p) q  \right)=    \HE\!\left(   q\, p  \right)\,,
\]
we easily infer Eq. \eqref{WOcl+}. 

To prove formula \eqref{Sbeta}, we note that  the adjoint of $W_\beta^+$ with respect to the $L^2(\RE^2)$ inner product is
$$
\big((W_\beta^{+})^{*}\, f\big)(q,p) = f(q, p) - \HE(-q\,p)\; \frac{f_{ev}(q, p)}{1 - {2i\,|p| \over m\,\ac}}\,.
$$
Then, by composition we get
\begin{align*}
\big(S_{\ac}^{cl} f\big)(q,p) = \big((W_\beta^{+})^{*} W_\beta^{-}\, f\big)(q,p)=\; & \big(W_\beta^{-} f\big)(q, p) - \HE(-q\,p)\; \frac{(W_\beta^{-} f)_{ev}(q, p)}{1 - {2i\,|p| \over m\,\ac}} \\
=\; & f(q, p) - \HE( q\,p)\; \frac{f_{ev}(q, p)}{1 - {2i\,|p| \over m\,\ac}} - \HE( - q\,p)\; \frac{f_{ev}(q, p) }{1 - {2i\,|p| \over m\,\ac}} \\
=\; & f(q, p) - \, \frac{f_{ev}(q, p)}{1 - {2i\,|p| \over m\,\ac}} \;.
\end{align*}
\end{proof}

\begin{remark}\label{r:Wpm-unitary} On account of Eq. \eqref{WOcl+}, it is easy to check that 
\[
W_\beta^{+} W_\beta^{-} f = W_\beta^{-} W_\beta^{+} f\,.
\]
Moreover, from the identity 
\[\frac1{1 + {2i\,|p| \over m\,\ac}} + \frac1{1 - {2i\,|p| \over m\,\ac}}  =\frac{2}{ \big|1 + {2i\,|p| \over m\,\ac}\big|^{2}}
\] 
 and a straightforward calculation it follows that 
\[
W_\beta^{\pm} (W_\beta^{\pm})^{*} f = (W_\beta^{\pm})^{*} W_\beta^{\pm} f = f\,.
\]
Hence, in particular, $ S^{cl}_\beta W_\beta^{+} f =  W_\beta^- f$. 
\end{remark}
\begin{remark} By arguments similar to those used in the proof of Proposition \ref{prop-wave-cl}, one gets that the limits 
$$
\breve W^{\pm}_{\beta}f:=\lim_{t\to\pm\infty}e^{   i tL_{\beta}}e^{ - i  t L_{0}}f
$$ 
exist in $L^{2}(\RE^2)$ for any $f\in L^{2}(\RE^2)$ and
$$
\big(\breve W^{\pm}_{\beta}f\big)(p,q) = f(q, p) - \; \frac{\HE(\mp q\,p)}{1 \mp {2i\,|p| \over m\,\ac}}\  f_{ev}(q, p)\,.
$$
Therefore, by \cite[Ch. X, Th. 3.5]{Kato}, both $W^{\pm}_{\beta}$ and $\breve W^{\pm}_{\beta}$ are complete, and the absolutely continuous part of $L_{\beta}$ is unitarily equivalent to the absolutely continuous part of $L_{0}$, i.e., to  $L_{0}$ itself; thus $\sigma_{ac}(L_{\beta})=\sigma_{ac}(L_{0})=\RE$ and $L_{\beta}$ is unitarily equivalent to $L_{0}$.
\end{remark}

\section{Semi-classical limit of quantum Hamiltonians with delta potentials}
In this section we prove Theorems \ref{t:1} and \ref{t:2}. We start by recalling the definition and the basic properties  of the self-adjoint Hamiltonian $\H_\alpha$ in $L^2(\RE)$ corresponding to the  formal expression
$$
\H_\alpha := - \,{\hbar^2 \over 2m}\,\frac{d^{2}\ }{dx^{2}} + \alpha\, \delta_0 \,,
$$
where $\alpha$ is a real parameter  with dimensions $(\mbox{mass}) \times (\mbox{length})^3 \times (\mbox{time})^{-2}$, and $\delta_0$ is the Dirac delta-distribution, for more details we refer to \cite[Ch. 8]{Teta}.  

%

\subsection{Quantum Hamiltonian with a delta potential\label{ss:Halpha}}
The self-adjoint operator  $\H_{\alpha}$ can be easily defined as a sum of quadratic forms. Alternatively, see \cite[Ch. I.3]{AGHH}, \cite[Ch. 8]{Teta}, it can be defined as a self-adjoint extension of the symmetric operator given by the restriction of the free Hamiltonian 
$$
\H_{0}:H^{2}(\RE)\subset L^{2}(\RE)\to L^{2}(\RE)\,,\quad \H_{0}:= - \,{\hbar^2 \over 2m}\,\frac{d^{2}\ }{dx^{2}}
$$
to $\C_{c}^{\infty}(\RE\backslash\{0\})$ (here $H^{2}(\RE)$ denotes the usual Sobolev space of order two, namely $H^{2}(\RE):=\{\psi\!\in\! L^{2}(\RE)\,|\,\psi''\!\in\! L^{2}(\RE)\}$\,):
\[
\dom(\H_\alpha) := \Big\{ \psi \in H^1(\RE) \cap H^2(\RE\backslash \{0\})\,\Big|\, \psi'(0_{+}) - \psi'(0_{-}) = \frac{2m\alpha}{\hbar^2}\,\psi(0)\Big\}\,,
\]
\[
\H_{\alpha} \psi(x) :=  - \frac{\hbar^2}{2m}\, \psi''(x) \qquad \mbox{for $x\neq 0$, $\psi \in \dom (\H_\alpha)$}\,.
\]

Next, let us recall that for $\alpha > 0$ the Hamiltonian $\H_{\alpha}$ has purely absolutely continuous spectrum $\sigma(\H_{\alpha}) \equiv \sigma_{ac}(\H_{\alpha}) = [0,\infty)$; in this case, a complete set of generalized eigenfunctions for $\H_{\alpha}$ is given by
\begin{equation}\label{vfiRdef}
\vfi_{k}^+(x) := {e^{i k x} \over \sqrt{2\pi}} + R_{+}(k)\, {e^{-i|k||x|} \over \sqrt{2\pi}}\,, \quad 
R_{+}(k) = -\,\frac{1}{1+i\,\frac{\hbar^2 |k|}{m\alpha}} \qquad (k \in \RE)\,.
\end{equation}
For later purposes we also define the functions 
\begin{equation}\label{vfiRdef-}
\vfi_{k}^-(x) := {e^{i k x} \over \sqrt{2\pi}} + R_{-}(k)\, {e^{i|k||x|} \over \sqrt{2\pi}}\,, \quad 
R_{-}(k) = -\,\frac{1}{1-i\,\frac{\hbar^2 |k|}{m\alpha}}\qquad
(k \in \RE)\,;
\end{equation}
this is also a complete set of generalized eigenfunctions of $\H_\alpha$ for $\alpha>0$. 

For $\alpha < 0$, the absolutely continuous spectrum of $\H_{\alpha}$ remains $\sigma_{ac}(\H_{\alpha}) = [0,\infty)$; again, this part of the spectrum is related to a set of generalized eigenfunctions like those in Eq. \eqref{vfiRdef} (or, equivalently, like those in Eq. \eqref{vfiRdef-}). In addition, $\H_{\alpha}$ possesses a proper eigenvalue $\lambda_\alpha < 0$; the explicit expressions for this eigenvalue and for the corresponding eigenfunction are, respectively,
$$ 
\lambda_\alpha := -\,{m \alpha^2 \over 2\hbar^2}\,, \qquad 
\vfi_{\alpha}(x) := {\sqrt{m\,|\alpha|} \over \hbar}\;e^{-{m |\alpha| \over \hbar^2}\,|x|}\,.
$$
Let us remark that $\vfi_{\alpha}$ is real-valued, positive and normalized in $L^2(\RE)$.

For any $\alpha \in \RE$, taking into account the above spectral decomposition of $\H_{\alpha}$, let us consider the bounded operators 
$$
\FF_{\pm} : L^2(\RE) \to L^2(\RE)\,, \qquad 
(\FF_{\pm}\, \psi)(k) := \int_{\RE}\! dx\; \overline{\vfi^\pm_k(x)}\; \psi(x) \,.
$$
Moreover we  introduce the orthogonal projectors
\begin{gather}
P_{ac} : L^2(\RE) \to L^2(\RE)\,, \qquad (P_{ac} \psi)(x) := \int_{\RE}\! dk\;\vfi_k^+(x)\; (\FF_{+}\, \psi)(k) \,, \label{defPac} \\
P_{\alpha} : L^2(\RE) \to L^2(\RE)\,, \qquad (P_{\alpha} \psi)(x) := \HE(-\alpha)\, \vfi_{\alpha}(x) \int_{\RE}\! dy\; \vfi_{\alpha}(y)\; \psi(y) \,; \label{defPpp}
\end{gather}
these fulfill
\begin{equation}
P_{ac} + P_{\alpha} = \uno\,. \label{orto}
\end{equation}

Eq.s \eqref{defPac} - \eqref{defPpp} reduce to
\begin{equation}
P_{ac} = \uno\,, \quad P_{{\alpha}} = 0 \qquad \mbox{for $\alpha > 0$}\,. \label{Pacuno}
\end{equation}

Taking into account the previously described spectral decomposition of $\H_{\alpha}$, for any $\alpha \!\in\! \RE$ the time evolution of any state $\psi \!\in\! L^2(\RE)$ induced by the unitary group $e^{- i {t \over \hbar} \H_\alpha}$ ($t \!\in \!\RE$) can be characterized as
\begin{equation}\label{qevol}
\big(e^{-i \frac{t}{\hbar} \H_\alpha}\, \psi \big)(x) = \int_{\RE}\! dk\; e^{-i {t \over \hbar}\frac{\hbar^2 k^2}{2m}}\; \vfi_{k}^+(x)\, (\FF_{+}\, \psi)(k) + e^{-i {t \over \hbar}\, \lambda_\alpha}\, (P_{\alpha}\psi)(x)\,.
\end{equation}

In the discussion above, in the definition of $P_{ac}$ and in Eq. \eqref{qevol}, one could equivalently use the generalized eigenfunctions $\vfi_k^-$ and the bounded operator $\FF_{-}$, respectively in place of $\vfi_k^+$ and $\FF_{+}$.  

\begin{remark}\label{remWpm} Existence and asymptotic completeness of the wave operators $\Omega_\alpha^\pm$ is well-known \cite{AGHH}. In particular, let us stress that $\Omega_\alpha^\pm$ are unitary on the absolutely continuous subspace for $\H_{\alpha}$, namely, on  $\ran (P_{ac})$ with $P_{ac}$  defined according to Eq. \eqref{defPac}; more precisely, there holds \cite[p. 531, Th. 3.2]{Kato}
\[
(\Omega_\alpha^{\pm})^{*} {\Omega}_\alpha^{\pm} = P_{ac}\,.
\]
We also recall that the wave operators have an explicit expression in terms of the bounded operators $\FF_{\pm}$ and of the Fourier transform \[
\big(\FF\psi\big)(k) \equiv \widehat\psi(k): =  {1 \over \sqrt{2\pi}} \int_{\RE}\!\!dx\;e^{-i k x}\,\psi(x)\,,
\]
i.e., 
\begin{equation}\label{st-rep}
\Omega_\alpha^\pm = \FF_{\pm}^{*} \FF \,.
\end{equation}
Relation \eqref{st-rep} is well known in the case of perturbations by regular potentials and can also be proved, by essentially the same kind of proof, in the case of a singular perturbation (see \cite[Th. 5.5]{JST}).
\end{remark}

\subsection{Semiclassical evolution of a coherent state in presence of a delta potential}

Following the approach of \cite{Hag1} (with respect to the notation employed therein, we set $\alpha=1/2$), we focus our attention on coherent states of the form
\begin{equation}\label{CohSt}
\psih(x) \equiv \psih(\sigmaq,\sigmap,q,p;x) = \frac{1}{(2\pi \hbar)^{1/4} \sqrt{\sigmaq}}\; e^{- \frac{\sigmap}{4\hbar\sigmaq} (x-q)^2 + i \frac{p}{\hbar} (x-q)} \qquad (x \in \RE)\,,
\end{equation}
where $(q,p) \in \RE^2$ and $\sigmaq,\sigmap \in \CO$ are such that
\begin{equation}\label{sxsp1}
\Re \sigmaq > 0\,, \qquad \Re \sigmap > 0\,, \qquad
\Re \!\big[\overline{\sigmaq}\,\sigmap\big] = 1\,.
\end{equation}
The assumption \eqref{sxsp1} grants that conditions (1.1)-(1.4) in \cite{Hag1} are satisfied setting $A=\sigmaq$ and $B=\sigmap$; in fact, we have $\Re \!\big[\sigmap\,\sigmaq^{-1} \big] \!=\! |\sigmaq|^{-2}$ and $\Re \!\big[\sigmaq\, \sigmap^{-1} \big] \!=\! |\sigmap|^{-2} \Re(\sigmaq \overline{\sigmap}) \!=\! |\sigmap|^{-2} \Re(\overline{\sigmaq} \sigmap) \!=\! |\sigmap|^{-2}$.
We remark that  the determination of the argument of complex numbers is such that square roots like $\sqrt{\sigmaq}$ fulfil $\Re \sqrt{\sigmaq} > 0$; this ensures, in particular, that $\sqrt{\overline{\sigmaq}\sigmap}\sqrt{\sigmaq} = |\sigmaq| \sqrt{\sigmap}$ for $\sigmaq$ and $\sigmap$ as in Eq. \eqref{sxsp1}.

Let us point out that the Fourier transform  with respect to $x$ of any state $\psih$ of the form \eqref{CohSt} reads
\begin{align}
\big(\FF\psih\big)(k)\equiv \psiht (\sigmaq,\sigmap,q,p;k) 
& = {1 \over \sqrt{\sigmap}} \left({2 \hbar \over \pi}\right)^{\!\!1/4} e^{-{\hbar \sigmaq \over \sigmap}(k- p/\hbar)^2 -  i k q}\,. \label{psih}
\end{align}

Moreover, all the states $\psih$ of the form \eqref{CohSt} fulfil the following relations (here $\widehat q$ and $\widehat p$ denote the position and momentum operators in the Schr\"odinger representation):
\begin{gather}
\|\psih\|_{L^2(\RE)} = 1\;;  \nonumber\\
\langle \widehat{q} \rangle_{\psih} := \int_\RE\!\! dx\;x \;|\psih(x)|^2 = q\;, \qquad
\langle \Delta  \widehat{q}\rangle_{\psih} := \sqrt{\langle \widehat{q}^2 \rangle_{\psih}\! - \langle \widehat{q} \rangle_{\psih}^2} = \sqrt{\hbar}\;|\sigmaq|\;; \label{meanq}\\
\langle \widehat{p} \rangle_{\psih} := \int_\RE dk\;(\hbar k)\;|\psiht(k)|^2 = p\;, \qquad
\langle \Delta \widehat{p} \rangle_{\psih} := \sqrt{\langle \widehat{p}^2 \rangle_{\psih}\! - \langle \widehat{p} \rangle_{\psih}^2} = \sqrt{\hbar}\;{|\sigmap| \over 2}\;. \label{meanp}
\end{gather}

In the sequel we analyze the time evolution, generated by the unitary group $ e^{-i \frac{t}{\hbar} \H_\alpha}$ ($t \in \RE$),  of an initial state of the form (see also Eq. \eqref{initial})
\begin{equation}\label{psi0}
 \psih_{\sigmazero,\xi}(x)  = \psih(\sigmazero,\sigmazero^{-1},q,p;x)
 \qquad \big(\xi = (q,p)\big)\,.
\end{equation}
The assigned parameters $(q,p) \!\in\! \RE^2$ and $\sigmazero \!>\! 0$ correspond respectively to the mean position, momentum and (rescaled) standard deviations of the initial state $\psih_{\sigmazero,\xi}$ (cf. Eq.s \eqref{meanq} and \eqref{meanp}). In this connection, let us remark that $\psih_{\sigmazero,\xi}$ saturates the uncertainty relation, i.e., $\langle \Delta \widehat q \rangle_{\psih_{\sigmazero,\xi}} \langle \Delta \widehat p \rangle_{\psih_{\sigmazero,\xi}} \!=  \hbar/ 2$.

To proceed let us recall  that the free evolution of the initial state $\psih$, determined by the unperturbed Hamiltonian operator $\H_0 := \hat{p}^2/(2m)$, is given by Eq.s \eqref{free} - \eqref{freeAsigmaq}. 

We also note that  the Fourier transform of $e^{- i {t \over \hbar} \H_0} \psih$ is given by
\[
\Big(\!\FF \big(e^{- i {t \over \hbar} \H_0}\, \psih(\sigmaq,\sigmap,q,p)\big)\!\Big)(k)  = e^{- i {t \over \hbar}\,{\hbar^2 k^2 \over 2m}}\,\psiht (\sigmaq,\sigmap,q,p;k)\,. 
\]

Moreover, let us remark that for any initial state $\psih_{\sigmazero,\xi}$ of the form \eqref{psi0} one has 
\[
\big(e^{- i {t \over \hbar} \H_0} \psih_{\sigmazero,\xi} \big)(x) = e^{{i \over \hbar}\, \SA_t}\; \psih\big(\sigma_t,\sigmazero^{-1},q_t,p;x\big) \qquad (t\in\RE)\,,  \]
with
\[
\SA_t := {p^2\,t \over 2m}\,, \qquad \sigma_t := \sigmazero + i\,{t \over 2m\sigmazero}\,, \qquad q_t := q + {p t \over m}\,.
\]

Finally, let us repeat that to simplify the exposition in the sequel we restrict the attention to coherent states fulfilling the condition $qp \neq 0$.

In the following proposition we obtain a convenient formula for the action of the unitary group $e^{-i \frac{t}{\hbar} \H_\alpha}$ on a coherent state.
\begin{proposition}\label{propSemi}
For any $\psih$ of the form \eqref{CohSt} with $qp\not=0$, there holds
\begin{align}
\big(e^{-i \frac{t}{\hbar} \H_\alpha} \psih \big)(x) & = \big(e^{-i \frac{t}{\hbar} \H_0} \psih \big)(x) + \HE(q p)\, F^{\hbar}_{+,t}\big(\!-\sgn(q) |x|\big) + \HE(-q p)\, F^{\hbar}_{-,t}\big(\!- \sgn(q) |x|\big) \nonumber \\
& \qquad + E^{\hbar}_{1,t}(x) + E^{\hbar}_{2,t}\big(x\big) + E^{\hbar}_{{\alpha},t}(x) \,, \label{psital}
\end{align}
where we set
\begin{align}
F^{\hbar}_{\pm,t}(x) \,:=\, & \,\frac1{\sqrt{2\pi}} \int_{\RE}\! dk\; e^{-i \frac{\hbar t}{2m}\, k^2} e^{i k x}\, R_{\pm}(k)\; \psiht(k)\,, 
\label{F+t}
\end{align}
\begin{equation}
E^{\hbar}_{1,t}(x) := \frac1{{2\pi}} \int_{\RE}\!\! dk\; e^{-i \frac{\hbar t}{2m}\, k^2} \Big(\!e^{ikx}\, R_{-}(k) + e^{-i|k||x|}\, |R_{+}(k)|^2\Big)\! \int_{\RE}\!\! dy\, \big(e^{i|k||y|} - e^{i \sgn(q) |k| y}\big)\, \psih (y)\,, \label{E1} \end{equation}
\begin{equation}\label{E2}
E^{\hbar}_{2,t}(x) := \frac{\sgn(q p)}{\sqrt{2\pi}} \int_{0}^{\infty}\!\! dk\; e^{-i \frac{\hbar t}{2m}\, k^2} e^{i \sgn(q p) k |x|}\, \Big[\,R_{-}(k) - R_{+}(k)\Big]\, \psiht\big(\!-\sgn(p) k\big)\,,
\end{equation}
and
\begin{equation}\label{Epp}
E^{\hbar}_{{\alpha},t}(x) := e^{-i {t \over \hbar}\, \lambda_\alpha}\, (P_{\alpha}\psi^{\hbar})(x)\,. 
\end{equation}
\end{proposition}
\begin{proof}
Taking into account the results of Section \ref{ss:Halpha} (see, in particular, Eq. \eqref{qevol}), for any $\psih$ of the form \eqref{CohSt} with $qp\not=0$,  we obtain
\begin{equation}\label{psit1}
\begin{aligned} 
\big(e^{-i \frac{t}{\hbar} \H_\alpha}\, \psih \big)(x) 
& = \big(e^{-i \frac{t}{\hbar} \H_0}\, \psih \big)(x) + \frac1{{2\pi}} \int_{\RE}\! dk\; e^{-i \frac{\hbar t}{2m}\, k^2}\, e^{ikx}\, R_{-}(k) \int_{\RE}\! dy\; e^{i|k||y|}\, \psih (y)  \\ 
& \qquad + \frac1{{2\pi}} \int_{\RE}\! dk\; e^{-i \frac{\hbar t}{2m}\, k^2}\, e^{-i|k||x|}\, R_{+}(k) \int_{\RE}\! dy\; e^{-iky}\, \psih (y) \\ 
& \qquad + \frac1{{2\pi}} \int_{\RE}\! dk\; e^{-i \frac{\hbar t}{2m}\, k^2}\, e^{-i|k||x|}\, |R_{+}(k)|^2 \int_{\RE}\! dy\; e^{i|k||y|}\, \psih (y) \\
& \qquad + e^{-i {t \over \hbar}\, \lambda_\alpha}\, (P_{\alpha}\psih)(x)\,.
\end{aligned}
\end{equation}
By the elementary identity 
$$
\int_{\RE}\! dy\; e^{i|k||y|}\, \psih (y) = \int_{\RE}\! dy\; e^{i \sgn(q) |k| y}\, \psih (y) + \int_{\RE}\! dy\; \big(e^{i|k||y|} - e^{i \sgn(q) |k| y}\big)\, \psih (y)\,,
$$
Eq. \eqref{psit1} is reformulated as follows:
\begin{align*} 
\big(e^{-i \frac{t}{\hbar} \H_\alpha}\, \psih \big)(x) 
& = \big(e^{-i \frac{t}{\hbar} \H_0}\, \psih \big)(x) + \frac1{\sqrt{2 \pi}} \int_{\RE}\! dk\; e^{-i \frac{\hbar t}{2m}\, k^2}\, e^{ikx}\, R_{-}(k)\; \psiht\big(\!-\!\sgn(q) |k|\big) \nonumber \\ 
& \qquad + \frac1{\sqrt{2\pi}} \int_{\RE}\! dk\; e^{-i \frac{\hbar t}{2m}\, k^2}\, e^{-i|k||x|}\, R_{+}(k)\; \psiht(k) \nonumber \\ 
& \qquad + \frac1{\sqrt{2\pi}} \int_{\RE}\! dk\; e^{-i \frac{\hbar t}{2m}\, k^2}\, e^{-i|k||x|}\, |R_{+}(k)|^2 \; \psiht\big(\!-\!\sgn(q) |k|\big) \nonumber \\
& \qquad + E^{\hbar}_{1,t}(x) + E^{\hbar}_{{\alpha},t}(x)\,.
\end{align*}
Noting that $R_{\pm}(k) = R_{\pm}(-k)$, we obtain the following identities:
\begin{gather}
\frac1{\sqrt{2\pi}} \int_{\RE}\! dk\; e^{-i \frac{\hbar t}{2m}\, k^2} e^{ikx}\, R_{-}(k)\; \psiht\big(\!-\!\sgn(q) |k|\big) = \frac1{\sqrt{2\pi}} \int_{\RE}\! dk\; e^{-i \frac{\hbar t}{2m}\, k^2} e^{-ik|x|}\, R_{-}(k)\; \psiht\big(\!-\!\sgn(q) |k|\big)\,; \nonumber \\
\frac1{\sqrt{2\pi}} \int_{\RE}\! dk\; e^{-i \frac{\hbar t}{2m}\, k^2}\, e^{-i|k||x|}\, R_{+}(k)\; \psiht(k) = \frac1{\sqrt{2\pi}} \int_{\RE}\! dk\; e^{-i \frac{\hbar t}{2m}\, k^2}\, e^{-i|k||x|}\, R_{+}(k)\; \psiht\big(\!-\!\sgn(q) k\big)\,; \label{hey} \\
\frac1{\sqrt{2\pi}} \int_{\RE}\!\! dk\; e^{-i \frac{\hbar t}{2m}\, k^2} e^{-i|k||x|}\, |R_{+}(k)|^2\; \psiht\big(\!-\!\sgn(q) |k|\big) = \frac2{\sqrt{2\pi}} \int_{0}^\infty\!\!\! dk\; e^{-i \frac{\hbar t}{2m}\, k^2} e^{-ik|x|}\, |R_{+}(k)|^2\; \psiht\big(\!-\!\sgn(q) k\big)\,. \nonumber
\end{gather}
From the above relations we infer
\begin{align*} 
\big(e^{-i \frac{t}{\hbar} \H_\alpha}\, \psih \big)(x) 
& = \big(e^{-i \frac{t}{\hbar} \H_0}\, \psih \big)(x) 
+ \frac1{\sqrt{2\pi}} \int_{\RE}\! dk\; e^{-i \frac{\hbar t}{2m}\, k^2} e^{-ik|x|}\, R_{-}(k)\; \psiht\big(\!-\!\sgn(q) |k|\big) \nonumber \\ 
& \qquad + \frac1{\sqrt{2\pi}} \int_{\RE}\! dk\; e^{-i \frac{\hbar t}{2m}\, k^2} e^{-i|k||x|}\, R_{+}(k)\; \psiht\big(\!-\!\sgn(q) k \big) \nonumber \\ 
& \qquad + \frac2{\sqrt{2\pi}} \int_{0}^\infty\!\!\! dk\; e^{-i \frac{\hbar t}{2m}\, k^2} e^{-ik|x|}\, |R_{+}(k)|^2\; \psiht\big(\!-\!\sgn(q) k\big) \nonumber \\
& \qquad + E^{\hbar}_{1,t}(x) + E^{\hbar}_{{\alpha},t}(x)\,.
\end{align*}

Hence, taking into account the basic identity 
\begin{equation} \label{RReq}
R_{+}(k) +  R_{-}(k) = 2\,\Re R_{+}(k) = - \,2\,|R_{+}(k)|^2 \qquad \mbox{for all $k \in \RE$}\;,
\end{equation}
we obtain
\begin{align*}
\big(e^{-i \frac{t}{\hbar} \H_\alpha}\, \psih \big)(x) 
& = \big(e^{-i \frac{t}{\hbar} \H_0}\, \psih \big)(x) + \frac1{\sqrt{2\pi}} \int_{-\infty}^{0}\! dk\; e^{-i \frac{\hbar t}{2m}\, k^2} e^{-ik|x|}\, R_{-}(k)\; \psiht\big(\sgn(q) k\big) \nonumber \\ 
& \qquad + \frac1{\sqrt{2\pi}} \int_{-\infty}^{0}\! dk\; e^{-i \frac{\hbar t}{2m}\, k^2} e^{i k |x|}\, R_{+}(k)\; \psiht\big(\!-\!\sgn(q) k\big) \nonumber \\
& \qquad + E^{\hbar}_{1,t}(x) + E^{\hbar}_{{\alpha},t}(x) \\
& = \big(e^{-i \frac{t}{\hbar} \H_0}\, \psih \big)(x) + \frac1{\sqrt{2\pi}} \int_{0}^{\infty}\! dk\; e^{-i \frac{\hbar t}{2m}\, k^2} e^{ik|x|}\, R_{-}(k)\; \psiht\big(\!-\!\sgn(q) k\big) \nonumber \\ 
& \qquad + \frac1{\sqrt{2\pi}} \int_{0}^{\infty}\! dk\; e^{-i \frac{\hbar t}{2m}\, k^2} e^{-i k |x|}\, R_{+}(k)\; \psiht\big(\sgn(q) k\big) \nonumber \\
& \qquad + E^{\hbar}_{1,t}(x) + E^{\hbar}_{{\alpha},t}(x)\,.
\end{align*}

Next, let us note the chain of identities 
\begin{equation}\label{right}
 \HE\big(\sgn(q) k\big)  = \HE(q ) - \sgn(q)\, \HE\big(\!- k\big) = \HE(q p) - \sgn(q p)\,  \HE\big(\!-\sgn(p) k\big)\,.
\end{equation}
The first identity is trivial for $q<0$, while for $q>0$ it is equivalent to the   identity $ \HE(k)  = 1 - \HE(- k)$. The second identity  is trivial for $p>0$, while for $p<0$ it follows from $\HE(q p) - \sgn(q p)\,\HE\big(\!-\sgn(p) k\big) = \HE(-q) - \sgn(-q)\, \HE(k) = \HE\big(\sgn(-q)(-k)\big)= \HE(\sgn(q) k)$. 
By Eq. \eqref{right}, we immediately infer
\begin{align*}
& \HE\big(\sgn(q) k\big) R_{+}(k) + \HE\big(\!-\sgn(q) k\big) R_{-}(k) \\
& \hspace{2.5cm} = \HE(q p)\,R_{+}(k) + \HE(-q p)\,R_{-}(k) + \sgn(q p)\, \HE\big(\!-\sgn(p) k\big) \Big[R_{-}(k) - R_{+}(k)\Big]\,.
\end{align*}
Using this identity, by a few additional elementary manipulations we obtain
\begin{align*}
\big(e^{-i \frac{t}{\hbar} \H_\alpha}\, \psih \big)(x) 
& = \big(e^{-i \frac{t}{\hbar} \H_0}\, \psih \big)(x) + E^{\hbar}_{1,t}(x) + E^{\hbar}_{{\alpha},t}(x) \\
& \qquad + \frac1{\sqrt{2\pi}} \int_{\RE}\! dk\; e^{-i \frac{\hbar t}{2m}\, k^2} e^{ik|x|}\, \HE(k)\;R_{-}(k)\; \psiht\big(\!-\!\sgn(q) k\big) \nonumber \\ 
& \qquad + \frac1{\sqrt{2\pi}} \int_{\RE}\! dk\; e^{-i \frac{\hbar t}{2m}\, k^2} e^{-i k |x|}\, \HE(k)\;R_{+}(k)\; \psiht\big(\sgn(q)k\big) \\
& = \big(e^{-i \frac{t}{\hbar} \H_0}\, \psih \big)(x) + E^{\hbar}_{1,t}(x) + E^{\hbar}_{{\alpha},t}(x) \\
& \qquad + \frac1{\sqrt{2\pi}} \int_{\RE}\! dk\; e^{-i \frac{\hbar t}{2m}\, k^2} e^{-i \sgn(q) k |x|}\, \HE\big(\!-\!\sgn(q) k\big)\;R_{-}(k)\; \psiht(k) \nonumber \\ 
& \qquad + \frac1{\sqrt{2\pi}} \int_{\RE}\! dk\; e^{-i \frac{\hbar t}{2m}\, k^2} e^{-i \sgn(q) k |x|}\, \HE\big(\sgn(q) k\big)\;R_{+}(k)\; \psiht(k) \\
& = \big(e^{-i \frac{t}{\hbar} \H_0}\, \psi \big)(x) + E^{\hbar}_{1,t}(x) + E^{\hbar}_{{\alpha},t}(x) \\
& \qquad + \frac{\HE(q p)}{\sqrt{2\pi}} \int_{\RE}\! dk\; e^{-i \frac{\hbar t}{2m}\, k^2} e^{-i \sgn(q) k |x|}\; R_{+}(k)\; \psiht(k) \\
& \qquad + \frac{\HE(-q p)}{\sqrt{2\pi}} \int_{\RE}\! dk\; e^{-i \frac{\hbar t}{2m}\, k^2} e^{-i \sgn(q) k |x|}\; R_{-}(k)\; \psiht(k) \\
& \qquad + \frac{\sgn(q p)}{\sqrt{2\pi}} \int_{\RE}\! dk\; e^{-i \frac{\hbar t}{2m}\, k^2} e^{-i \sgn(q) k |x|}\, \HE\big(\!-\sgn(p) k\big)\,\Big[\,R_{-}(k) - R_{+}(k)\Big]\, \psiht(k)\,.
\end{align*}
The proof is concluded noting that the latter identity is equivalent to Eq. \eqref{psital}. 
\end{proof}

Let us recall that $\psiht(k)$ is a Gaussian state concentrated in a neighbourhood of  $k = p$; with this in mind, we proceed to analyse the dominant contributions in Eq. \eqref{psital}. To this purpose we state and prove the following:

\begin{lemma}\label{LemmaEpbis}  There exists a constant $C>0$ such that, for any $\psih$ of the form \eqref{CohSt} with $qp\not=0$,  for all $t\in\RE$ and for all $\ve >0$, there holds 
\begin{equation}\label{boundEpbis}
\big\|F^{\hbar}_{\pm,t} - R_{\pm}(p/\hbar)\;e^{-i \frac{t}{\hbar} \H_0}\, \psih \big\|_{L^2(\RE)} \leq C \left[\!\left(\! \hbar\Big({|\sigmap| \over m|\alpha|}\Big)^{\!2/3}\right)^{\!\!3/2-\ve}\! +  \,e^{-{1 \over 2} \Big(\!{1 \over \hbar} \big(\!{m|\alpha| \over |\sigmap|}\!\big)^{\!2/3}\!\Big)^{\!\!2\ve}}\right] .
\end{equation}
\end{lemma}
\begin{proof}We prove the bound for $F^{\hbar}_{+,t}$; the proof for $F^{\hbar}_{-,t}$ is identical and therefore omitted.

Let us notice that by unitarity of the Fourier transform we have
\begin{align*}
& \big\|F^{\hbar}_{+,t} - R_{+}(p/\hbar)\;e^{-i \frac{t}{\hbar} \H_0}\, \psih \big\|_{L^2(\RE)}
 = \left\|{1 \over \sqrt{2\pi}}\int_{\RE}\! dk\; e^{-i \frac{\hbar t}{2m}\, k^2} e^{i k x} \big[R_{+}(k) - R_{+}(p/\hbar)\big]\, \psiht(k) \right\|_{L^2(\RE)}\\ 
& = \left\|e^{-i \frac{\hbar t}{2m}\, k^2} \big[R_{+}(k) - R_{+}(p/\hbar)\big]\, \psiht(k) \right\|_{L^2(\RE,dk)}= \left(\int_{\RE}\! dk\; \big|R_{+}(k) - R_{+}(p/\hbar)\big|^2\, \big|\psiht(k)\big|^2\right)^{1/2}\,.
\end{align*}
Recalling the definition of $R_{+}(k)$ given in Eq. \eqref{vfiRdef}, by explicit computations we obtain ({\footnote{In particular note that, for all $b \in \RE $, there holds $$ \sup_{\xi \in \RE}\,{b^2\, (|\xi| - 1)^{2} \over (1 \!+\! b^2) (1 \!+\! b^2 \xi^2)} = {b^2 \over 1 \!+\! b^2}\; \max\left\{\left.{ (|\xi| \!-\! 1)^{2} \over 1 \!+\! b^2 \xi^2}\right|_{\xi = 0}, \left.{(|\xi| \!-\! 1)^{2} \over 1 \!+\! b^2 \xi^2}\right|_{\xi = \infty} \right\} = {\max\{b^2, 1\} \over 1 + b^2} \leq 1\,. $$}})
\begin{align}
\big|R_{+}(k) - R_{+}(p/\hbar)\big|^2 
& = {\big(\frac{\hbar^2}{m\alpha}\big)^{\!2}\, \big(|k| - |p/\hbar| \big)^{2} \over \left(1 + \big({\hbar p \over m \alpha}\big)^{\!2}\right)\! \left(1 + \big({\hbar^2 k \over m \alpha}\big)^{\!2}\right)} \nonumber \\
& \leq \min\left\{1\,,\Big(\frac{\hbar^2}{m\alpha}\Big)^{\!2} \big(\,|k| - |p/\hbar|\, \big)^{2}\right\}\,. \label{inequalityR}
\end{align}

Next, let us fix $\ve > 0$ and note that 
\begin{align*}
\big\|F^{\hbar}_{+,t} - R_{+}(p/\hbar)\;e^{-i \frac{t}{\hbar} \H_0}\, \psih \big\|_{L^2(\RE)}
 \leq\, & \,\bigg(\int_{\left\{|\,k - p/\hbar\,| \,\leq\, {|\sigmap| \over \sqrt{\hbar}} \Big(\!{1 \over \hbar} \big(\!{m|\alpha| \over |\sigmap|}\!\big)^{\!2/3}\!\Big)^{\!\!\ve}\right\}} \hspace{-0.1cm} dk\; \big|R_{+}(k) - R_{+}(p/\hbar)\big|^2\, \big|\psiht(k)\big|^2\bigg)^{\!\!1/2} \\ 
 & +  \bigg(\int_{\left\{|\,k - p/\hbar\,| \,\geq\, {|\sigmap| \over \sqrt{\hbar}} \Big(\!{1 \over \hbar} \big(\!{m|\alpha| \over |\sigmap|}\!\big)^{\!2/3}\!\Big)^{\!\!\ve}\right\}} \hspace{-0.1cm} dk\; \big|R_{+}(k) - R_{+}(p/\hbar)\big|^2\, \big|\psiht(k)\big|^2\bigg)^{\!\!1/2}.
\end{align*}
On the one hand, taking into account the inequality \eqref{inequalityR} and the elementary relation $\big||k|-|p/\hbar|\big| \leq |k - p/\hbar|$, we infer
\begin{align*}
& \int_{\left\{|\,k - p/\hbar\,| \,\leq\, {|\sigmap| \over \sqrt{\hbar}} \Big(\!{1 \over \hbar} \big(\!{m|\alpha| \over |\sigmap|}\!\big)^{\!2/3}\!\Big)^{\!\!\ve}\right\}} \hspace{-0.1cm} dk\; \big|R_{+}(k) - R_{+}(p/\hbar)\big|^2\, \big|\psiht(k)\big|^2 \\
& \leq \left(\frac{\hbar^2}{m\alpha}\right)^{\!\!2} \int_{\left\{|\,k - p/\hbar\,| \,\leq\, {|\sigmap| \over \sqrt{\hbar}} \Big(\!{1 \over \hbar} \big(\!{m|\alpha| \over |\sigmap|}\!\big)^{\!2/3}\!\Big)^{\!\!\ve}\right\}} \hspace{-0.1cm} dk\; \big(|k| - |p/\hbar| \big)^{2}\, \big|\psiht(k)\big|^2 \\
& \leq  \left(\frac{\hbar^2}{m\alpha}\right)^{\!\!2} \left({|\sigmap| \over \sqrt{\hbar}} \left(\! {1 \over \hbar}\left(\!{m|\alpha| \over |\sigmap|}\!\right)^{\!\!2/3}\right)^{\!\!\ve}\right)^{\!\!2} \int_{\RE}\! dk\; \big|\psiht(k)\big|^2 
= \left(\!\hbar \left(\frac{|\sigmap|}{m|\alpha|}\right)^{\!\!2/3}\right)^{\!\!3-2\ve}\,.
\end{align*}

On the other hand, we have
\begin{align*}
& \int_{\left\{|\,k - p/\hbar\,| \,\geq\, {|\sigmap| \over \sqrt{\hbar}} \Big(\!{1 \over \hbar} \big(\!{m|\alpha| \over |\sigmap|}\!\big)^{\!2/3}\!\Big)^{\!\!\ve}\right\}} \hspace{-0.1cm} dk\; \big|R_{+}(k) - R_{+}(p/\hbar)\big|^2\, \big|\psiht(k)\big|^2
\leq \int_{\left\{|\,k - p/\hbar\,| \,\leq\, {|\sigmap| \over \sqrt{\hbar}} \Big(\!{1 \over \hbar} \big(\!{m|\alpha| \over |\sigmap|}\!\big)^{\!2/3}\!\Big)^{\!\!\ve}\right\}} \hspace{-0.1cm} dk\; \big|\psiht(k)\big|^2 \\
& = {1 \over |\sigmap|}\, \sqrt{{2 \hbar \over \pi}} \int_{\left\{|\,k - p/\hbar\,| \,\leq\, {|\sigmap| \over \sqrt{\hbar}} \Big(\!{1 \over \hbar} \big(\!{m|\alpha| \over |\sigmap|}\!\big)^{\!2/3}\!\Big)^{\!\!\ve}\right\}} \hspace{-0.1cm} dk\;e^{- {2 \hbar (k - p/\hbar)^2 \over |\sigmap|^2}} \\
& \leq {1 \over |\sigmap|}\, \sqrt{{2 \hbar \over \pi}}\;\, e^{- \Big(\!{1 \over \hbar} \big(\!{m|\alpha| \over |\sigmap|}\!\big)^{\!2/3}\!\Big)^{\!\!2\ve}} \int_{\RE}\! dk\; e^{- {\hbar (k - p/\hbar)^2 \over |\sigmap|^2}} 
= \sqrt2 \;e^{- \Big(\!{1 \over \hbar} \big(\!{m|\alpha| \over |\sigmap|}\!\big)^{\!2/3}\!\Big)^{\!\!2\ve}} \,.
\end{align*}
Summing up, the above arguments and the basic relation $\sqrt{a^2 + b^2} \leq a + b$ for $a,b\geq 0$ yield the bound \eqref{boundEpbis}. 
\end{proof}

Next we analyze the remainders in Eq. \eqref{psital}. 

In the following lemma we show that the projection of a coherent state on the eigenspace corresponding to the  (negative) eigenvalue of $\H_{\alpha}$ (for $\alpha < 0$) is small for $\hbar$ ``small enough''. Hence, in the semiclassical limit,  the presence of bound states (states corresponding to negative eigenvalues) has negligible effects on  the time evolution,  generated by the unitary group $ e^{-i \frac{t}{\hbar} \H_\alpha}$, of a coherent state. In other words, the  $L^2$-norm of $E^{\hbar}_{{\alpha},t}$ in Eq. \eqref{psital} is small  in the semiclassical limit. 

\begin{lemma}\label{LemmaProj} Let $P_{ac}$ and $P_{\alpha}$ be defined, respectively, as in Eq.s \eqref{defPac} and \eqref{defPpp}. Then, there exists a constant $C > 0$ such that, for any $\psih$ of the form \eqref{CohSt} with $qp\not=0$,  there holds
\begin{equation}
\big\|P_{ac} \psih - \psih\big\|_{L^2(\RE)} \leq C\, \Big( e^{- {m |\alpha|\,|q| \over 8 \hbar^2}} + e^{- \frac{q^2}{8 \hbar|\sigmaq|^2}}\Big)\, ; \label{Pac}
\end{equation}
equivalently,
\begin{equation}
\big\|P_{\alpha} \psih \big\|_{L^2(\RE)} \leq 
C\, \Big( e^{- {m |\alpha|\,|q| \over 8 \hbar^2}} + e^{- \frac{q^2}{8 \hbar|\sigmaq|^2}}\Big) \,. \label{Ppp}
\end{equation}
\end{lemma}
\begin{proof} For $\alpha > 0$, Eq.s \eqref{Pac} - \eqref{Ppp} follow trivially from the identities reported in Eq. \eqref{Pacuno}.

Let us now assume $\alpha < 0$. Hereafter we show how to derive Eq. \eqref{Ppp}; due to Identity \eqref{orto}, this suffices to infer Eq. \eqref{Pac} as well. Consider the definition \eqref{defPpp} of $P_{\alpha}$ and recall that $\|\vfi_{\alpha}\|_{L^2(\RE)} = 1$; then, for any $\eta_1 \in (0,1)$ we get
\begin{align*}
& \big\|P_{\alpha} \psih\big\|_{L^2(\RE)} = \left|\int_{\RE}\! dx\; \vfi_{\alpha}(x)\; \psih(x)\right| \leq \int_{\RE}\! dx\; \vfi_{\alpha}(x)\; \big|\psih(x)\big| \\
& = \frac{\sqrt{m\,|\alpha|}}{\hbar\, (2\pi \hbar)^{1/4} \sqrt{|\sigmaq|}} \left( \int_{|x-q| \,\leq\, \eta_1 |q|}\!\!\! dx\; e^{-{m |\alpha| \over \hbar^2}\,|x|}\; e^{- \frac{(x-q)^2}{4\hbar|\sigmaq|^2}} + \int_{|x-q| \,>\, \eta_1 |q|}\!\!\! dx\; e^{-{m |\alpha| \over \hbar^2}\,|x|}\; e^{- \frac{(x-q)^2}{4\hbar|\sigmaq|^2}} \right) .
\end{align*}
On the one hand, noting that $|x| \geq |q| - |x-q| \geq (1-\eta_1)\, |q|$ for $|x-q| \leq \eta_1\,|q|$ and using the Cauchy-Schwarz inequality, for any $\eta_2 \in (0,1)$ we obtain
\begin{align*}
& \frac{\sqrt{m\,|\alpha|}}{\hbar\, (2\pi \hbar)^{1/4} \sqrt{|\sigmaq|}} \int_{|x-q| \,\leq\, \eta_1 |q|}\!\!\! dx\; e^{-{m |\alpha| \over \hbar^2}\,|x|}\; e^{- \frac{(x-q)^2}{4\hbar|\sigmaq|^2}} \\
& \leq e^{- \eta_2\,(1-\eta_1)\, {m |\alpha| \over \hbar^2}\,|q|} \int_{|x-q| \,\leq\, \eta_1 |q|}\!\!\! dx \left(\frac{\sqrt{m\,|\alpha|}}{\hbar}\;e^{-(1-\eta_2){m |\alpha| \over \hbar^2}\,|x|}\right)\! \left(\frac{1}{(2\pi \hbar)^{1/4} \sqrt{|\sigmaq|}}\; e^{- \frac{(x-q)^2}{4\hbar|\sigmaq|^2}}\right) \\
& \leq e^{- \eta_2\,(1-\eta_1)\, {m |\alpha| \over \hbar^2}\,|q|} \left(\frac{m\,|\alpha|}{\hbar^2} \int_{\RE}\! dx\; e^{-(1-\eta_2){2 m |\alpha| \over \hbar^2}\,|x|}\right)^{\!\!1/2}\! \left(\frac{1}{\sqrt{2\pi \hbar}\; |\sigmaq|} \int_{\RE}\! dx'\; e^{- \frac{(x\!'-q)^2}{2\hbar|\sigmaq|^2}}\right)^{\!\!1/2} \\
& = {1 \over \sqrt{1-\eta_2}}\; e^{- \eta_2\,(1-\eta_1)\, {m |\alpha| \over \hbar^2}\,|q|}\;.
\end{align*}
On the other hand, for any $\eta_3 \in (0,1)$ by elementary arguments we infer 
\begin{align*}
& \frac{\sqrt{m\,|\alpha|}}{\hbar\, (2\pi \hbar)^{1/4} \sqrt{|\sigmaq|}} \int_{|x-q| \,\geq\, \eta_1 |q|}\!\!\! dx\; e^{-{m |\alpha| \over \hbar^2}\,|x|}\; e^{- \frac{(x-q)^2}{4\hbar|\sigmaq|^2}} \\
& \leq e^{- \eta_3\, \eta_1^2\, \frac{q^2}{4\hbar|\sigmaq|^2}}  \int_{|x-q| \,\geq\, \eta_1 |q|}\!\!\! dx \left(\frac{\sqrt{m\,|\alpha|}}{\hbar}\;e^{-{m |\alpha| \over \hbar^2}\,|x|}\right)\! \left(\frac{1}{(2\pi \hbar)^{1/4} \sqrt{|\sigmaq|}}\; e^{- (1-\eta_3) \frac{(x-q)^2}{4\hbar|\sigmaq|^2}}\right) \\
& \leq e^{- \eta_3\, \eta_1^2\, \frac{q^2}{4\hbar|\sigmaq|^2}} \left(\frac{m\,|\alpha|}{\hbar^2} \int_{\RE}\! dx\;e^{-{2 m |\alpha| \over \hbar^2}\,|x|}\right)^{\!\!1/2} \left(\frac{1}{\sqrt{2\pi \hbar}\; |\sigmaq|} \int_{\RE}\! dx' \; e^{-(1-\eta_3) \frac{(x\!'-q)^2}{2\hbar|\sigmaq|^2}}\right)^{\!\!1/2} \\
& =  {1 \over (1-\eta_3)^{1/4}}\;e^{- \eta_3\, \eta_1^2\, \frac{q^2}{4\hbar|\sigmaq|^2}}\;.
\end{align*}
With the specific admissible choices $\eta_1 = \eta_3 = 2^{-1/3}$ and $\eta_2 = 2^{-8/3}(2^{1/3} - 1)^{-1}$, the above estimates imply Eq. \eqref{Ppp}, whence the thesis.
\end{proof}

\begin{lemma}\label{LemmaE1bis} 
 There exists a constant $C>0$ such that, for any $\psih$ of the form \eqref{CohSt} with $qp\not=0$,  and for all $t \in \RE$, there holds
\[
\big\|E^{\hbar}_{1,t}\big\|_{L^2(\RE)} \leq C e^{- {q^2 \over 4\hbar|\sigmaq|^2}}\,.
\]
\end{lemma}

\begin{proof} Firstly, let us remark that the definition \eqref{E1} of $E_{1,t}$ can be reformulated as follows, recalling the basic Identity \eqref{RReq}:
\begin{align*}
E^{\hbar}_{1,t}(x) 
& = \frac1{{2\pi}} \int_{0}^{\infty}\!\!\! dk\; e^{-i \frac{\hbar t}{2m}\, k^2} \Big((e^{ikx} \!+\! e^{-ikx})\, R_{-}(k) + e^{-i k|x|}\, 2\,|R_{+}(k)|^2\Big) \int_{\RE}\! dy\, \big(e^{i k |y|} - e^{i \sgn(q) k y}\big)\, \psih (y) \\
& = \frac1{{2\pi}} \int_{0}^{\infty}\!\!\! dk\; e^{-i \frac{\hbar t}{2m}\, k^2} \Big(e^{i k |x|}\, R_{-}(k) - e^{-i k |x|}\,R_{+}(k)\Big) \int_{\RE}\! dy\, \big(e^{i k |y|} - e^{i \sgn(q) k y}\big)\, \psih (y)\,.
\end{align*}
To proceed, notice that
$$
e^{i k |x|}\, R_{-}(k) - e^{-i k |x|}\,R_{+}(k) = -\,{e^{i k |x|} \over 1 - i\,{\hbar^2 k \over m \alpha}} + {e^{-i k |x|} \over 1 + i\,{\hbar^2 k \over m \alpha}} \qquad \mbox{for $k > 0$}\,,
$$
and that the expression on the r.h.s. is an odd function of $k$, for $k \in \RE$; moreover, we have
\begin{align*}
& \int_{\RE}\!\! dy\, \big(e^{i k |y|} - e^{i \sgn(q) k y}\big)\, \psih (y) 
= \int_{0}^{\infty}\!\! dy\, \big(e^{i k y} - e^{- i k y}\big)\, \psih\big(\!-\!\sgn(q) y\big)\,,
\end{align*}
which is also an odd function of $k \in \RE$. Thus, by symmetry arguments we obtain
\begin{align*}
E^{\hbar}_{1,t}(x)
& = \frac1{{4\pi}} \int_{\RE}\!\! dk\; e^{-i \frac{\hbar t}{2m}\, k^2}\! \left(\!-\,{e^{i k |x|} \over 1 - i\,{\hbar^2 k \over m \alpha}} + {e^{-i k |x|} \over 1 + i\,{\hbar^2 k \over m \alpha}}\right)\! \int_{0}^{\infty}\!\! dy\, \big(e^{i k y} - e^{- i k y}\big)\, \psih\big(\!-\!\sgn(q) y\big) \\
& = -\,\frac1{{2\pi}} \int_{\RE}\!\! dk\; e^{-i \frac{\hbar t}{2m}\, k^2} {e^{i k |x|} \over 1 - i\,{\hbar^2 k \over m \alpha}} \int_{0}^{\infty}\!\! dy\, \big(e^{i k y} - e^{- i k y}\big)\, \psih\big(\!-\!\sgn(q) y\big)\,.
\end{align*}
By the elementary inequality $\| \psi (|\,\cdot\,|)\|_{L^2(\RE)}^2 \leq 2\,\| \psi \|_{L^2(\RE)}^2$ and by the unitarity of the Fourier transform it follows that 
\begin{align*}
& \|E^{\hbar}_{1,t}\|_{L^2(\RE)}^2 \\ 
& \leq 2 \left\| \frac1{\sqrt{2\pi}}  {1 \over 1 - i\,{\hbar^2 k \over m \alpha}} \int_{0}^{\infty}\!\! dy\, \big(e^{i k y} - e^{- i k y}\big)\, \psih\big(\!-\!\sgn(q) y\big)\right\|_{L^2(\RE,dk)}^2 \\ 
& = \frac2{\pi}  \int_{\RE}\!\! dk\; {1 \over 1 \!+\!{\hbar^4 k^2 \over m^2 \alpha^2}} \int_{0}^{\infty}\!\! dy\,\int_{0}^{\infty}\!\! dy'\, \Big(\!\cos\big(k(y-y')\big) - \cos\big(k(y+y')\big)\!\Big)\, \psih\big(\!-\!\sgn(q) y\big)\, \overline{\psih\big(\!-\!\sgn(q) y'\big)}\; .
\end{align*}

Then, using the identity (see \cite[p. 424, Eq. 3.723.2]{GR})
\[
 \int_{\RE}\!\! dk\; {\cos(k\xi ) \over 1 + {\hbar^4 k^2 \over m^2 \alpha^2}} =  { m^2 \alpha^2 \over \hbar^4} \int_{\RE}\!\! dk\; {\cos(k\xi ) \over { m^2 \alpha^2 \over \hbar^4} \!+\! k^2}  = \pi\;  { m |\alpha| \over \hbar^2}\: e^{- { m |\alpha| \over \hbar^2}|\xi|}\;,
\]
we obtain 
\begin{align*}
 \|E^{\hbar}_{1,t}\|_{L^2(\RE)}^2  & \leq 2\;{ m |\alpha| \over \hbar^2}  \int_{0}^{\infty}\!\! dy\,\int_{0}^{\infty}\!\! dy'\, e^{- { m |\alpha| \over \hbar^2}|y-y'|}\,\psih\big(\!-\!\sgn(q) y\big)\, \overline{\psih\big(\!-\!\sgn(q) y'\big)}  \\ 
&\qquad -  2\;{ m |\alpha| \over \hbar^2}     \int_{0}^{\infty}\!\! dy\,\int_{0}^{\infty}\!\! dy'\,  e^{- { m |\alpha| \over \hbar^2}(y+y')}\, \psih\big(\!-\!\sgn(q) y\big)\, \overline{\psih\big(\!-\!\sgn(q) y'\big)}  \\
& =   \mathcal{I}^{\hbar}_{1} +  \mathcal{J}^{\hbar}_{1} +  \mathcal{K}^{\hbar}_{1} 
\end{align*}
where we put
\begin{align*}
 \mathcal{I}^{\hbar}_{1} &  :=  2\;{ m |\alpha| \over \hbar^2}  \int_{0}^{\infty}\!\! dy\,  \psih\big(\!-\!\sgn(q) y\big)  e^{ { m |\alpha| \over \hbar^2}y} \int_{y}^{\infty}\!\! dy'\, e^{- { m |\alpha| \over \hbar^2}y'}\,\overline{\psih\big(\!-\!\sgn(q) y'\big)}\;,  \\ 
  \mathcal{J}^{\hbar}_{1} & := 2\;{ m |\alpha| \over \hbar^2}  \int_{0}^{\infty}\!\! dy\,\psih\big(\!-\!\sgn(q) y\big) e^{- { m |\alpha| \over \hbar^2}y} \int_{0}^{y }\!\! dy'\, e^{ { m |\alpha| \over \hbar^2}y'}\, \overline{\psih\big(\!-\!\sgn(q) y'\big)} \;, \\ 
\mathcal{K}^{\hbar}_{1} & := -  2\;{ m |\alpha| \over \hbar^2}  \left|   \int_{0}^{\infty}\!\! dy  e^{- { m |\alpha| \over \hbar^2}y}\, \psih\big(\!-\!\sgn(q) y\big) \right|^2  \;.
\end{align*}

Keeping in mind the basic identity (cf. Eq. \eqref{CohSt} and the related comments)
$$
\big|\psih\big(\!-\!\sgn(q) y\big)\big| = \frac{1}{(2\pi \hbar)^{1/4} \sqrt{|\sigmaq|}}\; e^{- {(y+|q|)^2 \over 4\hbar|\sigmaq|^2}}\,,
$$
by elementary arguments we infer the following inequalities:
\begin{align*}
 \big|\mathcal{I}^{\hbar}_{1}\big| 
& \leq 2\;  { m |\alpha| \over \hbar^2}  \int_{0}^{\infty}\!\! dy\,  \big|\psih\big(\!-\!\sgn(q) y\big) \, \big| e^{ { m |\alpha| \over \hbar^2}y} \int_{y}^{\infty}\!\! dy'\, e^{- { m |\alpha| \over \hbar^2}y'}\, \big|\psih\big(\!-\!\sgn(q) y'\big)\big| \\ 
& \leq 2\; \frac{e^{- {q^2 \over 4\hbar|\sigmaq|^2}}}{(2\pi \hbar)^{1/4} \sqrt{|\sigmaq|}}  \int_{0}^{\infty}\!\! dy\,  \big|\psih\big(\!-\!\sgn(q) y\big) \big|\;   { m |\alpha| \over \hbar^2}\; e^{ { m |\alpha| \over \hbar^2}y} \int_{y}^{\infty}\!\! dy'\, e^{- { m |\alpha| \over \hbar^2}y'} \\
& \leq 2\; \frac{e^{- {q^2 \over 4\hbar|\sigmaq|^2}}}{(2\pi \hbar)^{1/4} \sqrt{|\sigmaq|}}  \int_{0}^{\infty}\!\! dy\,  \big|\psih\big(\!-\!\sgn(q) y\big) \big| \leq 2\; \frac{e^{- {q^2 \over 2\hbar|\sigmaq|^2}}}{(2\pi \hbar)^{1/2} {|\sigmaq|}}  \int_{0}^{\infty}\!\! dy\,  e^{- {y^2 \over 4\hbar|\sigmaq|^2}} \leq \sqrt 2\;e^{- {q^2 \over 2\hbar|\sigmaq|^2}}\;;
\end{align*}
\begin{align*} 
 | \mathcal{J}^{\hbar}_{1} | & \leq  2\;{ m |\alpha| \over \hbar^2}  \int_{0}^{\infty}\!\! dy\,\big| \psih\big(\!-\!\sgn(q) y\big) \big|\; e^{- { m |\alpha| \over \hbar^2}y} \int_{0}^{y }\!\! dy'\, e^{ { m |\alpha| \over \hbar^2}y'}\, \big|\psih\big(\!-\!\sgn(q) y'\big)\big|  \\
  & \leq 2\; \frac{e^{- {q^2 \over 4\hbar|\sigmaq|^2}}}{(2\pi \hbar)^{1/4} \sqrt{|\sigmaq|}}  
  \int_{0}^{\infty}\!\! dy\,\big| \psih\big(\!-\!\sgn(q) y\big) \big|\; {m |\alpha| \over \hbar^2}\; e^{- { m |\alpha| \over \hbar^2}y} \int_{0}^{y }\!\! dy'\, e^{ { m |\alpha| \over \hbar^2}y'} 
 \\
& \leq  2\;  \frac{e^{- {q^2 \over 4\hbar|\sigmaq|^2}}}{(2\pi \hbar)^{1/4} \sqrt{|\sigmaq|}}  \int_{0}^{\infty}\!\! dy\,  \big|\psih\big(\!-\!\sgn(q) y\big) \big|\leq  \sqrt 2\; e^{- {q^2 \over 2\hbar|\sigmaq|^2}}\;;
\end{align*}
\begin{align*}
|\mathcal{K}^{\hbar}_{1}| \leq  2\;{ m |\alpha| \over \hbar^2}  \left(\int_{0}^{\infty}\!\! dy \; e^{-2 { m |\alpha| \over \hbar^2}y}\right)\left(
   \int_{0}^{\infty}\!\! dy'\big| \psih\big(\!-\!\sgn(q) y'\big)\big|^2  \right) \leq \frac{e^{- {q^2 \over 2\hbar|\sigmaq|^2}}}{2}\;.
\end{align*}
Summing up, the above relations imply the thesis. 
\end{proof}

\begin{lemma}\label{LemmaE2bis}
There exists a constant $C>0$ such that, for any $\psih$ of the form \eqref{CohSt} with $qp\not=0$,  for all $t\in\RE$ and for all $\lambda > 0$, there holds
\begin{equation*}
\big\|E^{\hbar}_{2,t}\big\|_{L^2(\RE)}  \leq C\;e^{-{p^2 \over \hbar |\sigmap|^2}}\! \left[\left(\hbar\, \Big(\frac{|\sigmap|}{m\alpha}\Big)^{\!2/3}\right)^{\!\!{3 \over 2}-{3\lambda \over 2}} + e^{-{1 \over 2}\left(\!{1 \over \hbar} \big(\!{m|\alpha| \over |\sigmap|}\!\big)^{\!2/3}\right)^{\!2\lambda}} \right] .
\end{equation*}
\end{lemma}
\begin{proof}
Recalling the definition of $E^{\hbar}_{2,t}$ (see Eq.\! \eqref{E2}), by the elementary inequality $\| \psi (\pm |\,\cdot\,|)\|_{L^2(\RE)}^2 \!\leq\! 2\,\| \psi \|_{L^2(\RE)}^2$ and by unitarity of the Fourier transform, we have
\begin{align*}
\|E^{\hbar}_{2,t}\|_{L^2(\RE)}^2
& \leq  2 \left\|\frac1{\sqrt{2\pi}} \int_{\RE}\! dk\; e^{-i \frac{\hbar t}{2m}\, k^2}\,e^{i \sgn(q p) k x}\;\HE(k) \Big[R_{-}(k) - R_{+}(k)\Big] \psiht\big(\!- \sgn(p) k\big)\right\|_{L^2(\RE)}^2 \\
& = 2 \int_{0}^{\infty}\!\!dk\; \big|R_{-}(k) - R_{+}(k)\big|^{2}\, \big|\psiht\big(\!- \sgn(p) k\big)\big|^2\,.
\end{align*}
Moreover, from the definitions of $R_{\pm}(k)$ (see Eq.s \eqref{vfiRdef} and \eqref{vfiRdef-}), by explicit computations we obtain ({\footnote{In particular, note that $ \sup_{\xi \in \RE} [{\xi^{2}/ (1 \!+\! \xi^2)^2}] = 1/4$\,.}})
\begin{align*}
\big|R_{-}(k) - R_{+}(k)\big|^2 = {\big({\hbar^2 k \over m \alpha}\big)^2 \over \big(1 + \big({\hbar^2 k \over m \alpha}\big)^2\big)^2} \leq \min\left\{{1 \over 4}\,,\Big(\frac{\hbar^2 k}{m\alpha}\Big)^{\!2}\right\}\,.
\end{align*}

Thus, taking into account Identity \eqref{psih} for $\psiht$, we infer the following for any fixed $\lambda > 0$:
\begin{align*}
\big\|E^{\hbar}_{2,t}\big\|_{L^2(\RE)}^2 & \leq {8 \over |\sigmap|} \sqrt{{2 \hbar \over \pi}} \int_{0}^{\infty}\!\!dk\; \big|R_{-}(k) - R_{+}(k)\big|^2\; e^{-{2\hbar  (k + |p|/\hbar)^2 \over |\sigmap|^2}} \\
& \leq {8 \over |\sigmap|} \sqrt{{2 \hbar \over \pi}}\; e^{-{2\, p^2 \over \hbar |\sigmap|^2}} \int_{0}^{\infty}\!\!dk\;\big|R_{+}(k) - R_{-}(k)\big|^2\; e^{-{2\hbar k^2 \over |\sigmap|^2}} \\
& \leq {4 \over |\sigmap|} \sqrt{{2 \hbar \over \pi}}\; e^{-{2\, p^2 \over \hbar |\sigmap|^2}}\! \left[ \int_{\left\{|k| \,\leq\, {|\sigmap| \over \sqrt{\hbar}}\left(\!{1 \over \hbar} \big(\!{m|\alpha| \over |\sigmap|}\!\big)^{\!2/3}\right)^{\!\lambda} \right\}}\!dk\; \Big(\frac{\hbar^2 k}{m\alpha}\Big)^{\!2}\,e^{-{2\hbar k^2 \over |\sigmap|^2}}
\!+ {1 \over 4} \int_{\left\{|k| \,\geq\, {|\sigmap| \over \sqrt{\hbar}}\left(\!{1 \over \hbar} \big(\!{m|\alpha| \over |\sigmap|}\!\big)^{\!2/3}\right)^{\!\lambda} \right\}}\!dk\; e^{-{2\hbar k^2 \over |\sigmap|^2}} \right].
\end{align*}
On the one hand, we have
\begin{align*}
\int_{\left\{|k| \,\leq\, {|\sigmap| \over \sqrt{\hbar}}\left(\!{1 \over \hbar} \big(\!{m|\alpha| \over |\sigmap|}\!\big)^{\!2/3}\right)^{\!\lambda} \right\}}\!dk\, \Big(\frac{\hbar^2 k}{m\alpha}\Big)^{\!2}\, e^{-{2\hbar k^2 \over |\sigmap|^2}} 
& \leq \Big(\frac{\hbar^2}{m\alpha}\Big)^{\!2} \int_{\left\{|k| \,\leq\, {|\sigmap| \over \sqrt{\hbar}}\left(\!{1 \over \hbar} \big(\!{m|\alpha| \over |\sigmap|}\!\big)^{\!2/3}\right)^{\!\lambda} \right\}}\!dk\; k^2 \\
& = {2 \over 3} \Big(\frac{\hbar^2}{m\alpha}\Big)^{\!2}\! \left({|\sigmap| \over \sqrt{\hbar}}\Big({1 \over \hbar} \Big({m|\alpha| \over |\sigmap|}\Big)^{\!2/3}\Big)^{\!\lambda}\right)^{\!\!3}.
\end{align*}
On the other hand,
\begin{align*}
\int_{\left\{|k| \,\geq\, {|\sigmap| \over \sqrt{\hbar}}\left(\!{1 \over \hbar} \big(\!{m|\alpha| \over |\sigmap|}\!\big)^{\!2/3}\right)^{\!\lambda} \right\}}\!dk\; e^{-{2 \hbar k^2 \over |\sigmap|^2}}
\leq e^{-\left(\!{1 \over \hbar} \big(\!{m|\alpha| \over |\sigmap|}\!\big)^{\!2/3}\right)^{\!2\lambda}}\! \int_{\RE}\!dk\; e^{-{\hbar k^2 \over |\sigmap|^2}}
= \sqrt{\pi \over \hbar}\;|\sigmap|\; e^{-\left(\!{1 \over \hbar} \big(\!{m|\alpha| \over |\sigmap|}\!\big)^{\!2/3}\right)^{\!2\lambda}}\,.
\end{align*}

Summing up, the above relations imply
\[
\big\|E^{\hbar}_{2,t}\big\|_{L^2(\RE)}^2  \leq e^{-{2\, p^2 \over \hbar |\sigmap|^2}}\! \left[ {8 \over 3}\sqrt{{2 \over \pi}}\, \left(\hbar \Big(\frac{|\sigmap|}{m\alpha}\Big)^{\!2/3}\right)^{\!\!3-3\lambda} + \sqrt{2}\; e^{-\left(\!{1 \over \hbar} \big(\!{m|\alpha| \over |\sigmap|}\!\big)^{\!2/3}\right)^{\!2\lambda}} \right] ,
\]
which yields the thesis in view of the basic relation $\sqrt{a^2 + b^2} \leq a + b$ for $a,b\geq 0$.
\end{proof}

In the next lemma we collect all the results of the previous lemmata. 
\begin{lemma}\label{l:main} There exists a constant $C>0$ such that for any $\psih$ of the form \eqref{CohSt} with $qp\not=0$,  for all $t\in\RE$,  and for all $\ve_1,\ve_2>0 $, there holds
\begin{equation}\label{timedependent1}
\begin{aligned}
& \big\|e^{-i \frac{t}{\hbar} \H_\alpha} \psih - \upsilonh_t \big\|_{L^2(\RE)}  \\ 
& \leq  
 C\Bigg[\!\left(\! \hbar\Big({|\sigmap| \over m|\alpha|}\Big)^{\!2/3}\right)^{\!\!\frac32-\ve_1}\! +  \,e^{-{1 \over 2} \Big(\!{1 \over \hbar} \big(\!{m|\alpha| \over |\sigmap|}\!\big)^{\!2/3}\!\Big)^{\!\!2\ve_1}} 
 + e^{-{p^2 \over \hbar |\sigmap|^2}}\! \bigg(\left(\hbar\, \Big(\frac{|\sigmap|}{m|\alpha|}\Big)^{\!2/3}\right)^{\!\!{3 \over 2}-{\frac32 \lambda_2 }} + e^{-{1 \over 2}\left(\!{1 \over \hbar} \big(\!{m|\alpha| \over |\sigmap|}\!\big)^{\!2/3}\right)^{\!2\lambda_2}} \bigg)  \\
& \qquad\; +  e^{- {q^2 \over 8\hbar|\sigmaq|^2}} + e^{- {m |\alpha|\,|q| \over 8 \hbar^2}}\!\Bigg] 
\end{aligned}
\end{equation}
where 
\begin{equation}
\begin{aligned}\label{phiht2}
 \upsilonh_t (x) & := \big( e^{-i \frac{t}{\hbar} \H_0} \psih \big)(x)  + \HE(qp)\, R_{+}(p/\hbar)\;\big(e^{-i \frac{t}{\hbar} \H_0}\psih\big)\big(\!-\sgn(q)|x|\big)\\ 
& \qquad + \HE(-qp)\, {R_{-}(p/\hbar)}\;\big(e^{-i \frac{t}{\hbar} \H_0}\psih\big)\big(\!-\sgn(q)|x|\big)\,.
\end{aligned}
\end{equation}
\end{lemma}
\begin{proof} Note that, by Eq. \eqref{Ppp} of Lemma \ref{LemmaProj}, it follows that 
$$ \big\|E^{\hbar}_{{\alpha},t}\big\|_{L^2(\RE)} \leq  C \left( e^{- {m |\alpha|\,|q| \over 8 \hbar^2}} + e^{- \frac{q^2}{8 \hbar|\sigmaq|^2}}\right) .
$$
Then, claim \eqref{timedependent1} follows immediately from Eq. \eqref{psital}, together with  the expansions of the terms $F^{\hbar}_{\pm,t}$ in Lemma \ref{LemmaEpbis}, and  the bounds on the remainders $E^{\hbar}_{1,t}$, $E^{\hbar}_{2,t}$ in Lemmata \ref{LemmaE1bis}, \ref{LemmaE2bis}. In the bound \eqref{timedependent1} we also used
the trivial inequality  $e^{- {q^2 \over 4\hbar|\sigmaq|^2}} \leq e^{- {q^2 \over 8\hbar|\sigmaq|^2}}$.
\end{proof}


For later reference, let us further notice the following
\begin{lemma}\label{Lemmapsi0} For any $\psih$ of the form \eqref{CohSt} with $qp\not=0$,  there holds
\begin{gather}
\left\|\psih\big(\sigmaq,\sigmap,q,p; \sgn (q)|\cdot|\big) - \Big(\psih(\sigmaq,\sigmap,q,p;\,\cdot\,)+ \psih(\sigmaq,\sigmap,-q,-p;\,\cdot\,)\!\Big) \right\|_{L^2(\RE)} \leq   e^{- \frac{q^2}{4\hbar|\sigmaq|^2}}\,,\label{psi0mod1} \\
\left\|\psih\big(\sigmaq,\sigmap,q,p;- \sgn (q)|\cdot|\big) \right\|_{L^2(\RE)} \leq  e^{- \frac{q^2}{4\hbar|\sigmaq|^2}}\,. \label{psi0mod2}
\end{gather}
\end{lemma}
\begin{proof} As a example, we give the proof for $q \geq 0$. Similar arguments can be employed for $q < 0$, and we omit them for brevity.

Noting that $\psih(\sigmaq,\sigmap,-q,-p;x) = \psih(\sigmaq,\sigmap,q,p;-x)$, by direct computations we get
\begin{align*}
& \left\|\psih\big(\sigmaq,\sigmap,q,p; |\cdot|\big) - \Big(\psih(\sigmaq,\sigmap,q,p;\,\cdot\,)+ \psih(\sigmaq,\sigmap,-q,-p;\,\cdot\,)\!\Big) \right\|_{L^2(\RE)}^2 \\
& = \int_{\RE}\! dx\; \big|\HE(-x)\,\psih(\sigmaq,\sigmap,p,q;x) + \HE(x)\,\psih(\sigmaq,\sigmap,p,q;-x)\big|^2 \\
& = 2 \int_{\RE}\! dx\; \big|\HE(x)\,\psih(\sigmaq,\sigmap,p,q;-x) \big|^2 =  \frac{2}{\sqrt{2\pi \hbar}\; |\sigmaq|} \int_{0}^{\infty}\!\! dx\; e^{- \frac{(-x - q)^2}{2\hbar|\sigmaq|^2}} \\
& \leq \frac{2\,e^{- \frac{q^2}{2\hbar|\sigmaq|^2}}}{\sqrt{2\pi \hbar}\; |\sigmaq|} \int_{0}^{\infty}\!\! dx\; e^{- \frac{x^2}{2\hbar|\sigmaq|^2}} =  e^{- \frac{q^2}{2\hbar|\sigmaq|^2}}\,,
\end{align*}
where we used the identity $\Re(\sigmap/\sigmaq) = |\sigmaq|^{-2}$ (see Eq. \eqref{sxsp1}) and the inequality $e^{-(a+b)^2} \leq e^{-a^2-b^2}$ for $a,b\geq 0 $. This proves Eq. \eqref{psi0mod1} for $q \geq 0$.

On the other hand, we have
\[\begin{aligned}
& \left\|\psih\big(\sigmaq,\sigmap,p,q;-|\cdot|\big)\right\|_{L^2(\RE)}^2 = \frac{1}{\sqrt{2\pi \hbar}\; |\sigmaq|} \int_{\RE} dx \; e^{- \frac{(-|x|-q)^2}{2\hbar|\sigmaq|^2}}  \\ 
& \leq \frac{e^{- \frac{q^2}{2\hbar|\sigmaq|^2}}}{\sqrt{2\pi \hbar}\; | \sigmaq|} \int_{\RE} dx \; e^{- \frac{x^2}{2\hbar|\sigmaq|^2}} 
=  e^{- \frac{q^2}{2\hbar|\sigmaq|^2}}\,;
\end{aligned}
\]
this yields Eq. \eqref{psi0mod2} for $q \geq 0$, thus concluding the proof.
\end{proof}


\subsection{Proof of Theorem \ref{t:1}\label{ss:3.1}}

At first, in the following proposition we give an explicit formula for the semiclassical limit  evolution of a coherent state. 
\begin{proposition} \label{p:blacksun}
Let $\beta =  2\alpha/\hbar$. Then, under the assumptions of Theorem \ref{t:1} there holds
\begin{equation}\label{blacksun}
\begin{aligned}
e^{\frac{i}{\hbar}\Szero_{t}} \big(e^{itL_{\beta}}\phi^{\hbar}_{\sigma_{t},x}\big)(\xi)
=\, & 
\,\big(e^{-i\frac{t}{\hbar}\H_{0}}\,\psi^{\hbar}_{\sigmazero,\xi}\big)(x)  \\ 
& + \HE (-qp)\HE \!\left(t+\frac{m q}{p}\right) \, {R_{-}(p/\hbar)}\,\Big(
\big(e^{-i\frac{t}{\hbar}\H_{0}}\,\psi^{\hbar}_{\sigmazero,\xi}\big)(x) + \big(e^{-i\frac{t}{\hbar}\H_{0}}\,\psi^{\hbar}_{\sigmazero,\xi}\big)(-x)\Big)
\\ 
& +  \HE (qp)\HE \!\left(-t-\frac{m q}{p}\right) \, {R}_{+}(p/\hbar)\,\Big(
\big(e^{-i\frac{t}{\hbar}\H_{0}}\,\psi^{\hbar}_{\sigmazero,\xi}\big)(x) + \big(e^{-i\frac{t}{\hbar}\H_{0}}\,\psi^{\hbar}_{\sigmazero,\xi}\big)(-x)\Big)\,.
 \end{aligned}
 \end{equation}
\end{proposition}
\begin{proof}
Recall that $\xi = (q,p)$. We start by noticing that, by Eq. \eqref{Lbeta-group} (see also Eq. \eqref{ftal1} with $t\to -t$), 
\[
\big(e^{itL_{\beta}}\phi^{\hbar}_{\sigma_{t},x}\big)(\xi)
= 
\big(e^{itL_{0}}\phi_{\sigma_{t},x}\big)(q,p)-\frac{\HE (-tqp)\HE \!\left(\frac{|pt|}{m}-|q|\right)}{1-\sgn(t)\,\frac{2i|p|}{m\beta}}\Big(\big(e^{itL_{0}} \phi_{\sigma_{t},x}\big)(q,p)+\big(e^{itL_{0}}\phi_{\sigma_{t},x}\big)(-q,-p)\Big)\,;
\]
hence, on account of Identity \eqref{free2}, we infer 
\[
e^{\frac{i}{\hbar}\Szero_{t}} \big(e^{itL_{\beta}}\phi^{\hbar}_{\sigma_{t},x}\big)(\xi)
= 
\big(e^{-i\frac{t}{\hbar}\H_{0}}\,\psi^{\hbar}_{\sigmazero,\xi}\big)(x)-\frac{\HE (-tqp)\HE \!\left(\frac{|pt|}{m}-|q|\right)}{1-\sgn(t)\,\frac{2i|p|}{m\beta}}\Big(
\big(e^{-i\frac{t}{\hbar}\H_{0}}\,\psi^{\hbar}_{\sigmazero,\xi}\big)(x) + \big(e^{-i\frac{t}{\hbar}\H_{0}}\,\psi^{\hbar}_{\sigmazero,-\xi}\big)(x)\Big)\,. 
\]
We note that 
\[\begin{aligned}
\big(e^{-i\frac{t}{\hbar}\H_{0}}\,\psi^{\hbar}_{\sigmazero,-\xi}\big)(x) 
=\, &\,  e^{\frac{i}{\hbar}\Szero_{t}} \psi^{\hbar}\Big(\sigmazero+\frac{it}{2m\sigmazero},\sigmazero^{-1},- q-\frac{pt}{m},-p;x\Big)  \\ 
=\, & \,e^{\frac{i}{\hbar}\Szero_{t}}\psi^{\hbar} \Big(\sigmazero+\frac{it}{2m\sigmazero},\sigmazero^{-1},q+\frac{pt}{m},p;-x\Big)  = \big(e^{-i\frac{t}{\hbar}\H_{0}}\,\psi^{\hbar}_{\sigmazero,\xi}\big)(-x) \,,
\end{aligned}
\]
whence, 
\[
e^{\frac{i}{\hbar}\Szero_{t}} \big(e^{itL_{\beta}}\phi^{\hbar}_{\sigma_{t},x}\big)(\xi)
= 
\big(e^{-i\frac{t}{\hbar}\H_{0}}\,\psi^{\hbar}_{\sigmazero,\xi}\big)(x)-\frac{\HE (-tqp)\HE \!\left(\frac{|pt|}{m}-|q|\right)}{1-\sgn(t)\,\frac{2i|p|}{m\beta}}\Big(
\big(e^{-i\frac{t}{\hbar}\H_{0}}\,\psi^{\hbar}_{\sigmazero,\xi}\big)(x) + \big(e^{-i\frac{t}{\hbar}\H_{0}}\,\psi^{\hbar}_{\sigmazero,\xi}\big)(-x)\Big)\,. 
\]
To conclude the proof we observe that 
\begin{align*}
\frac{\HE (-tqp)\,\HE \!\left(\frac{|pt|}{m}-|q|\right)}{1-\sgn(t)\,\frac{2i|p|}{m\beta}} 
=\, &\, \frac{\HE (t)\,\HE (-qp)\,\HE \!\left(\frac{|p|t}{m}-|q|\right)}{1-\sgn(t)\,\frac{2i|p|}{m\beta}} + \frac{\HE (-t)\,\HE (qp)\,\HE \!\left(-\frac{|p|t}{m}-|q|\right)}{1-\sgn(t)\,\frac{2i|p|}{m\beta}} \nonumber \\ 
=\, &\, \frac{\HE (-qp)\,\HE \!\left(t-\frac{m |q|}{|p|}\right)}{1-\frac{2i|p|}{m\beta}} + \frac{\HE (qp)\,\HE \!\left(-t-\frac{m|q|}{|p|}\right)}{1+ \frac{2i|p|}{m\beta}}
 \nonumber \\ 
=\, &\, \frac{\HE (-qp)\,\HE \!\left(t+\frac{m q}{p}\right)}{1-\frac{2i|p|}{m\beta}} + \frac{\HE (qp)\,\HE \!\left(-t-\frac{mq}{p}\right)}{1+\frac{2i|p|}{m\beta}}\,. 
\end{align*}
Setting $\beta =  2\alpha/\hbar$ and recalling the definitions of $R_{\pm}(k)$ (see Eq.s \eqref{vfiRdef} and \eqref{vfiRdef-}) we obtain Identity \eqref{blacksun}.  
\end{proof}

We are now ready to prove Theorem \ref{t:1}. 
\begin{proof}[Proof of Theorem \ref{t:1}] 
First, we use  Lemma \ref{l:main} to approximate the state $e^{-i\frac{t}{\hbar}\H_{\alpha}}\,\psi^{\hbar}_{\sigmazero,\xi}$ with $\upsilonh_{\sigmazero,\xi,t}$ defined according to Eq. \eqref{phiht2} by 
\[
\begin{aligned}
\upsilonh_{\sigmazero,\xi,t} (x) & := \big( e^{-i \frac{t}{\hbar} \H_0} \psih_{\sigmazero,\xi} \big)(x)  + \HE(qp)\, R_{+}(p/\hbar)\;\big(e^{-i \frac{t}{\hbar} \H_0}\psih_{\sigmazero,\xi}\big)\big(\!-\sgn(q)|x|\big)\\ 
&\qquad + \HE(-qp)\, {R_{-}(p/\hbar)}\;\big(e^{-i \frac{t}{\hbar} \H_0}\psih_{\sigmazero,\xi}\big)\big(\!-\sgn(q)|x|\big)\,.
\end{aligned}
\]

Next, we compare $\upsilonh_{\sigmazero,\xi,t}$ with the expression for $e^{\frac{i}{\hbar}\Szero_{t}}\,e^{itL_{\beta}} \phi^{\hbar}_{\sigma_{t},x}(\xi)$ from Proposition \ref{p:blacksun}.  We start by noticing that, if $t>-mq/p$ one has
\[\begin{aligned}
&q_t = q + \frac{pt}{m} > q - \frac{p}{m} \cdot \frac{mq}{p} =0 \quad & \text{for $p>0$}\,, \\
&q_t = q + \frac{pt}{m} < q - \frac{p}{m} \cdot \frac{mq}{p} =0 \quad & \text{for $p<0$}\,;
\end{aligned}
\]
hence, $\sgn(q_t) = \sgn (p)$. This implies, in turn,
\[\begin{aligned}
  \HE(q p) \big(e^{-i \frac{t}{\hbar} \H_0}\psih_{\sigmazero,\xi}\big)\big(\!-\sgn(q)|x|\big) 
& = \HE(q p)  \big(e^{-i \frac{t}{\hbar} \H_0}\psih_{\sigmazero,\xi}\big)\big(\!-\sgn(p)|x|\big) \\
& = \HE(q p) \big(e^{-i \frac{t}{\hbar} \H_0}\psih_{\sigmazero,\xi}\big)\big(\!-\sgn(q_t)|x|\big) \, 
\end{aligned}
\]
and, by a  similar argument, 
\[
\HE(-q p) \big(e^{-i \frac{t}{\hbar} \H_0}\psih_{\sigmazero,\xi}\big)\big(-\sgn(q)|x|\big) 
= \HE(-q p)  \big(e^{-i \frac{t}{\hbar} \H_0}\psih_{\sigmazero,\xi}\big)\big(\sgn(q_t)|x|\big) \, .
\]

On the other hand, for $t<- mq/p$ one has $\sgn(q_t) = - \sgn (p)$, hence  
\[
\HE(q p) \big(e^{-i \frac{t}{\hbar} \H_0}\psih_{\sigmazero,\xi}\big)\big(-\sgn(q)|x|\big)
= \HE(q p) \big(e^{-i \frac{t}{\hbar} \H_0}\psih_{\sigmazero,\xi}\big)\big(\sgn(q_t)|x|\big)\, 
\]
and
\[
\HE(-q p) \big(e^{-i \frac{t}{\hbar} \H_0}\psih_{\sigmazero,\xi}\big)\big(-\sgn(q)|x|\big)
= \HE(-q p) \big(e^{-i \frac{t}{\hbar} \H_0}\psih_{\sigmazero,\xi}\big)\big(-\sgn(q_t)|x|\big)\, .
\]
By the identities above we obtain the following formula for $\upsilonh_{\sigmazero,\xi,t} $, 
\[
\begin{aligned}
\upsilonh_{\sigmazero,\xi,t} (x) = \,&\,  \big( e^{-i \frac{t}{\hbar} \H_0} \psih_{\sigmazero,\xi} \big)(x)   \\ 
 &  +  \HE \!\left(t+\frac{m q}{p}\right) \HE(qp)  R_{+}(p/\hbar)\,  \big(e^{-i \frac{t}{\hbar} \H_0}\psih_{\sigmazero,\xi}\big)\big(\!-\sgn(q_t)|x|\big) \\
 & +  \HE \!\left(t+\frac{m q}{p}\right)  \HE(-qp)\, {R_{-}(p/\hbar)}  \big(e^{-i \frac{t}{\hbar} \H_0}\psih_{\sigmazero,\xi}\big)\big(\sgn(q_t)|x|\big) \\ 
  &  +  \HE \!\left(- t - \frac{m q}{p}\right) \HE(qp)  R_{+}(p/\hbar)\,  \big(e^{-i \frac{t}{\hbar} \H_0}\psih_{\sigmazero,\xi}\big)\big(\!\sgn(q_t)|x|\big) \\
 & +  \HE \!\left( - t - \frac{m q}{p}\right)  \HE(-qp)\, {R_{-}(p/\hbar)}  \big(e^{-i \frac{t}{\hbar} \H_0}\psih_{\sigmazero,\xi}\big)\big(- \sgn(q_t)|x|\big)\,.
\end{aligned}
\]

Recall that 
\[
\big(e^{-i\frac{t}{\hbar}\H_{0}}\,\psi^{\hbar}_{\sigmazero,\xi}\big)(x) =   e^{\frac{i}{\hbar}\Szero_{t}}\, \psi^{\hbar}\big(\sigma_t,\sigmazero^{-1},q_t,p;x\big)
\]
with $\sigma_t = \sigmazero+\frac{it}{2m\sigmazero} $ and $q_t =  q + \frac{pt}{m}$. Hence, by Lemma \ref{Lemmapsi0}, we infer:
\[
\left\|\big(e^{-i \frac{t}{\hbar} \H_0}\psih_{\sigmazero,\xi}\big)\big(\!\sgn(q_t)|\cdot|\big)- \Big(
\big(e^{-i\frac{t}{\hbar}\H_{0}}\,\psi^{\hbar}_{\sigmazero,\xi}\big)(x) + \big(e^{-i\frac{t}{\hbar}\H_{0}}\,\psi^{\hbar}_{\sigmazero,\xi}\big)(-x)\Big)\right\|_{L^2(\RE)} \leq e^{- \frac{q_t^2}{4\hbar|\sigma_t|^2}} 
\]
and 
\[
\left\|\big(e^{-i \frac{t}{\hbar} \H_0}\psih_{\sigmazero,\xi}\big)\big(\!-\sgn(q_t)|\cdot|\big)\right\|_{L^2(\RE)} \leq  e^{- \frac{q_t^2}{4\hbar|\sigma_t|^2}}\,. 
\]

The latter bounds, together with Proposition \ref{p:blacksun}, imply 
\[
\left\|     \upsilonh_{\sigmazero,\xi,t} (x)  - e^{\frac{i}{\hbar}\Szero_{t}} \big(e^{itL_{\beta}}\phi^{\hbar}_{\sigma_{t},(\cdot)}\big)(\xi)\right\|_{L^2(\RE)}   \leq 2\, e^{- \frac{q_t^2}{4\hbar|\sigma_t|^2}}\,. 
\]
Besides the exponential  $e^{- \frac{q_t^2}{4\hbar|\sigma_t|^2}}$, the remaining  terms on the r.h.s. of Eq. \eqref{t1} are a consequence of the fact that we approximated $e^{-i\frac{t}{\hbar}\H_{\alpha}}\,\psi^{\hbar}_{\sigmazero,\xi}$ with $\upsilonh_{\sigmazero,\xi,t}$ and of Lemma \ref{l:main}. 
\end{proof}
\subsection{Proof of Corollary \ref{c:1}.}\label{proof-coroll}
\begin{proof}
Choose $\lambda_1,\lambda_2$ in Theorem \ref{t:1} so that  $\lambda_1 = \frac32 \lambda_2 =\lambda$, and note that for $\lambda\in(0,3/2)$ the time-independent part on the r.h.s. of inequality \eqref{t1} is bounded by 
 \[
 C\left[ \underline h^{\frac32-\lambda}+  e^{-{1 \over 2 \underline h^{2\lambda}}} + e^{-{1 \over 2 \underline h^{4\lambda/3}}} + e^{- {1\over 8 \underline h }} + e^{- {1 \over 8 \underline h^2}}\right]  \leq C_*\, \underline h^{\frac32 - \lambda}\] 
 for some $C_*>0$, for all $\underline h < h_*$ with $h_*$ small enough.  

On the other hand, to take into account the time-dependent term on the r.h.s. of inequality \eqref{t1} it is enough to show that if $|t-t_{coll}(\xi)| \ge c_0\, |t_{coll}(\xi)|\,\sqrt{(3/2-\lambda)\,\underline h\,|\ln \underline h|}$ (for some $c_0 > 0$), then $ q_t^2/(4\hbar|\sigma_t|^2) \ge (3/2-\lambda)\,|\ln \underline h|$. Setting $y = 1 - t/t_{coll}(\xi)$, $a = {4\hbar\sigmazero^2 \over q^2}\,(3/2-\lambda)\,|\ln \underline h|$ and $b = {\hbar \over p^2\sigma_0^2}\,(3/2-\lambda)\,|\ln \underline h|$, the latter relation can be rephrased as $y^2/(a + b\, (1-y)^2) \ge 1$; a simple calculation shows that this inequality is fulfilled if
\begin{equation}\label{y}
a,b \in (0,1) \qquad \mbox{and} \qquad |y| \ge {b + \sqrt{a + b - a b} \over 1 - b}\,.
\end{equation}
Taking into account that $a + b - ab \le a + b$, $a \leq 4 (3/2-\lambda)\, \underline h\,|\ln \underline h|$ and $b \leq (3/2-\lambda)\, \underline h\,|\ln \underline h|$, it is easy to convince oneself that when $\underline h$ is small enough Eq. \eqref{y} holds true as soon as $|y| > c_0\,\sqrt{(3/2-\lambda)\,\underline h\,|\ln \underline h|}$ for some $c_0 > \sqrt{5}$, which proves Eq. \eqref{ttcoll}.
\end{proof}

\begin{remark}\label{r:dirichlet} Our main motivation for considering the classical dynamics generated by $L_{\ac}$ with $\ac =  {2 \alpha \over \hbar} \not= \infty$ is that it provides a better approximation for small $\hbar$ of the quantum dynamics induced by $\H_{\alpha}$ with $\alpha \not=\infty$, rather than the classical analogue corresponding to Dirichlet boundary conditions, i.e., to $L_{\infty} \equiv L_{\ac}$ with $\ac = \infty$.

More precisely, from Proposition \ref{exp-cl} (see also Remark \ref{RemDir}), Identity \eqref{free2} and the unitarity of $e^{-i {t \over \hbar} H_0}$ we infer the following, for $\xi \equiv (q,p)$:
\begin{align*}
& \left\|\big(e^{i t L_{\ac}} \phi_{\sigma_t,(\cdot)}^{\hbar}\big)(\xi) - \big(e^{i t L_{\infty}} \phi_{\sigma_t,(\cdot)}^{\hbar}\big)(\xi)\right\|_{L^2(\RE)}^2 \\
& = \HE(-t q p)\,\HE\!\left({|p t| \over m} \!-\! |q|\right)\!\left|{1 \over 1 - \sgn(t) {2 i |p| \over m \ac}} - 1\right|^2 \left\|\big(e^{i t L_{0}} \phi_{\sigma_t,(\cdot)}^{\hbar}\big)(\xi) + \big(e^{i t L_{0}} \phi_{\sigma_t,(\cdot)}^{\hbar}\big)(-\xi)\right\|_{L^2(\RE)}^2 \\
& = \HE(-t q p)\,\HE\!\left({|p t| \over m} \!-\! |q|\right)\!\left|{1 \over 1 - \sgn(t) {2 i |p| \over m \ac}} - 1\right|^2 \left\|e^{-i {t \over \hbar} H_{0}} \psih_{\sigma_0,\xi} + e^{-i {t \over \hbar} H_{0}} \psih_{\sigma_0,-\xi}\right\|_{L^2(\RE)}^2 \\
& = \HE(-t q p)\,\HE\!\left({|p t| \over m} \!-\! |q|\right) {\big({2 p \over m \ac}\big)^2 \over 1+ \big({2 p \over m \ac}\big)^2}\; \big\|\psih_{\sigma_0,\xi} + \psih_{\sigma_0,-\xi}\big\|_{L^2(\RE)}^2 \\
& = 2\;\HE(-t q p)\,\HE\!\left({|p t| \over m} \!-\! |q|\right) {\big({2 p \over m \ac}\big)^2 \over 1+ \big({2 p \over m \ac}\big)^2} \left(1 + e^{-{q^2 \over 2 \hbar |\sigma_0|^2}}\, e^{-{2 |\sigma_0|^2 p^2 \over \hbar}}\right)\,.
\end{align*}
Fixing $\ac =  {2 \alpha \over \hbar}$, for all $t \in \RE$ and for all ${\hbar |p| \over m \alpha} < 1$ the latter identity implies ({\footnote{Note that ${\xi^2 \over 1 + \xi^2} \geq {\xi^2 \over 2}$ for $|\xi| \leq 1$.}})
$$ \left\|\big(e^{i t L_{\ac}} \phi_{\sigma_t,(\cdot)}^{\hbar}\big)(\xi) - \big(e^{i t L_{\infty}} \phi_{\sigma_t,(\cdot)}^{\hbar}\big)(\xi)\right\|_{L^2(\RE)}^2 \geq \HE(-t q p)\,\HE\!\left({|p t| \over m} \!-\! |q|\right)\! \left({\hbar p \over m \alpha}\right)^{\!\!2}\,. $$ 
By the triangular inequality, from the above relation and from Theorem \ref{t:1} we infer that there exist two constants $C_*,\hbar_*>0$ such that for any $\psih_{\sigmazero,\xi}$ as in Eq. \eqref{psi0} and for all $\hbar \leq \hbar_*$ there holds
\begin{equation}
\big\|e^{-i \frac{t}{\hbar} \H_\alpha} \psih_{\sigmazero,\xi} - e^{{i \over \hbar}\, \Szero_t}\,e^{i t L_{\infty}} \psi_{\sigma_t,\xi}^{\hbar}\big\|_{L^2(\RE)} \geq C_*\; {\hbar |p| \over m |\alpha|} \label{timedependentDir}
\end{equation}
for $q p < 0$ and $t > - {m q \over p}$, or $q p > 0$ and $t < - {m q \over p}$.

Eq. \eqref{timedependentDir} makes patent that for small $\hbar$ the error in the approximation via classical Dirichlet boundary conditions is at least of order $\hbar$, while the analogous error corresponding to the classical dynamics induced by $L_{\ac}$ with $\ac =   {2 \alpha \over \hbar}$ is at most of order $\hbar^{\frac32- \lambda}$ for $\lambda \in (0,3/2)$.
\end{remark}

\subsection{Approximation of wave and scattering operators}

In the same spirit of the proof of Theorem \ref{t:1}, we look for a more convenient formula for the action of the wave operators on a coherent state.
\begin{proposition} For any $\psih$ of the form \eqref{CohSt} with $qp\not=0$,  there holds 
\begin{align}
\label{Omega_alpha+}
\Omega_\alpha^\pm\, \psih (x)=\, &  \,\psih (x) + \HE(q p)\, F^{\hbar}_{\pm,0}\big(\!\mp\sgn(q) |x|\big) + \HE(-q p)\, F^{\hbar}_{\pm,0}\big(\pm\sgn(q) |x|\big) +  E^{\hbar}_{3,\pm}(x)\,, 
\end{align}
with $F^{\hbar}_{\pm,0} \equiv F^{\hbar}_{\pm,t=0}$ defined according to Eq. \eqref{F+t} and
\[
E^{\hbar}_{3,\pm}(x) := \pm\,  \frac{\sgn(q p)}{\sqrt{2\pi}} \int_\RE\! dk\;  {R_{\pm}(k)} \Big( e^{-i\sgn(pq)k|x|}-  e^{i\sgn(pq)k|x|} \Big)  
\HE( k)\, \psiht(-\sgn(p)k)\,.
\]

\end{proposition} 
\begin{proof} As an example, we show how to derive Eq. \eqref{Omega_alpha+} for $\Omega_{\alpha}^{+}$; the proof for $\Omega_{\alpha}^{-}$ is similar and we omit it for brevity.

First recall that, by Remark \ref{remWpm}, 
\[
\Omega_\alpha^+\, \psih (x) 
 = \psih (x)  + \frac1{\sqrt{2\pi}} \int_{\RE}\! dk\;  e^{-i|k||x|}\, R_{+}(k)  \psiht (k) \,;
\]
from here and from Identity \eqref{hey} (to be used with $t=0$), we obtain 
\[\begin{aligned}
\Omega_\alpha^+\, \psih (x) & = \psih (x) 
+ \frac1{\sqrt{2\pi}} \int_{\RE}\! dk\;  e^{-i|k||x|}\, R_{+}(k)\; \psiht\big(\!-\!\sgn(q) k \big)  \\ 
& = \psih (x) + \frac1{\sqrt{2\pi}} \int_0^{\infty}\! dk\;  e^{-ik|x|}\, {R_{+}(k)}\; \psiht\big(\!-\! \sgn(q) k\big) \nonumber \\ 
& \qquad + \frac1{\sqrt{2\pi}} \int_{-\infty}^{0}\! dk\;  e^{i k |x|}\, R_{+}(k)\; \psiht\big(\!-\!\sgn(q) k\big) \nonumber \\
& = \psih(x) + \frac1{\sqrt{2\pi}} \int_\RE\! dk\; e^{-ik|x|}\, {R_{+}(k)} \HE (k) \Big( \psiht\big(\!-\!\sgn(q) k\big) +\psiht\big(\sgn(q) k\big) \Big)
\\
& =  \psih(x) + \frac1{\sqrt{2\pi}} \int_\RE\! dk\;  {R_{+}(k)} \Big(  e^{i\sgn(q)k|x|} \HE (-\sgn(q)k)+e^{-i\sgn(q)k|x|} \HE (\sgn(q)k)  \Big) \psiht(k)\,. 
\end{aligned}
\]
Next, from the identity $ \HE\big(\sgn(q) k\big)  = \HE(q p) - \sgn(q p)\,  \HE\big(\!-\sgn(p) k\big)$  (see Eq. \eqref{right}), we infer 
\[\begin{aligned}
\Omega_\alpha^+\, \psih (x)  = &   \psih(x) + \frac1{\sqrt{2\pi}} \int_\RE\! dk\;  {R_{+}(k)} \Big( e^{-i\sgn(q)k|x|} \HE(q p) +  e^{i\sgn(q)k|x|} \HE(- q p) \Big) \psiht(k) \\ 
&\qquad -   \frac1{\sqrt{2\pi}} \int_\RE\! dk\;  {R_{+}(k)} \Big( e^{-i\sgn(q)k|x|}-  e^{i\sgn(q)k|x|}  \Big) \sgn(q p) \,
\HE\big(\!-\sgn(p) k\big)\, \psiht(k)\,,
\end{aligned}
\]
which is equivalent to Eq. \eqref{Omega_alpha+}. 
\end{proof}

\begin{lemma}\label{LemmaE3} There exists a constant $C>0$ such that, for any $\psih$ of the form \eqref{CohSt} with $qp\not=0$,  there holds
\begin{equation}\label{boundE3+}
\big\|E^{\hbar}_{3,\pm}\big\|_{L^2(\RE)} \leq C\, e^{- {p^2 \over \hbar |\sigmap|^{2}}}\,.
\end{equation}
\end{lemma}
\begin{proof}
We prove the bound \eqref{boundE3+} for $E^{\hbar}_{3,+}$; the proof of the analogous bound for $E^{\hbar}_{3,-}$ is identical and we omit it for brevity. By the elementary inequality $\| \psi ( |\cdot|)\|_{L^2(\RE)}^2+\| \psi (-|\cdot|)\|_{L^2(\RE)}^2  \leq  4\| \psi \|_{L^2(\RE)}^2$, by unitarity of the Fourier transform, and by the bound $| R_{+}(k)|\leq 1$, we infer that 
\begin{align*}
\big\|E^{\hbar}_{3,+}\big\|_{L^2(\RE)}^2 
& \leq 4  \left\|\frac{1}{\sqrt{2\pi}} \int_\RE\! dk\;  e^{-ik x}\, \HE( k)\,R_{+}(k)\, \psiht(-\sgn(p)k) \right\|_{L^2(\RE)}^2 \\
& = 4  \int_{0}^{\infty}\!\!\!dk\; \big| R_{+}(k)\big|^{2}\, \big|\psiht\big(\!-\sgn(p) k\big)\big|^2 \leq 4 \int_{0}^{\infty}\!\!\!dk\; \big|\psiht\big(\!-\sgn(p) k\big)\big|^2 \\ 
& = {4 \over |\sigmap|} \left({2 \hbar \over \pi}\right)^{\!\!1/2} \int_{0}^{\infty}\!\!\!dk\; e^{-2 {\hbar (k + |p|/\hbar)^2 \over |\sigmap|^{2}}} 
\leq {4 \over |\sigmap|} \left({2 \hbar \over \pi}\right)^{\!\!1/2} e^{- {2 p^2 \over \hbar |\sigmap|^{2}}} \int_{0}^{\infty}\!\!\!dk\; e^{- {2 \hbar k^2 \over |\sigmap|^{2}}}
= 2 e^{- {2 p^2 \over \hbar |\sigmap|^{2}}}\,,
\end{align*}
which yields the thesis.
\end{proof}

\subsection{Proof of Theorem \ref{t:2}.}\label{proof2}
\begin{proof} We first prove claim \eqref{t2_1}. 

Preliminarily we apply the classical wave operators $W_\beta^\pm$, with $\beta =  2\alpha/\hbar$ to the state $\phi^{\hbar}_{\sigmazero,(\cdot)}(\xi)$. 
Recalling the definition of $R_{\pm}(k)$, by Eq. \eqref{WOcl+} we have:
\be\label{W-phi}
\big(W_\beta^\pm \phi^{\hbar}_{\sigmazero,(\cdot)}\big)(\xi) = \phi^{\hbar}_{\sigmazero,(\cdot)}(\xi) + \HE( \mp q\,p)\; R_{\pm}(p/\hbar)\big(\phi^{\hbar}_{\sigmazero,(\cdot)}(\xi) + \phi^{\hbar}_{\sigmazero,(\cdot)}(-\xi)\big)\,;
\ee

On the other hand, by Eq. \eqref{Omega_alpha+}, and Lemmata \ref{LemmaEpbis}, \ref{Lemmapsi0} and \ref{LemmaE3} we infer
\begin{align*}
&\Big\|\Omega_\alpha^\pm\psih_{\sigmazero ,\xi }- \Big[\psih_{\sigmazero ,\xi }\! + \HE(\mp q p)\,R_{\pm}(p/\hbar)\, \big(\psih_{\sigmazero,\xi} \!+ \psih_{\sigmazero,-\xi}\big)
\Big] \Big\|_{L^2(\RE)} \\ 
&  \leq  C \left(e^{- \frac{q^2}{4\hbar\sigmazero^2}}+e^{-\frac{\sigmazero^2 p^2}{\hbar}} +\left(\! {\hbar \over (m|\alpha|\sigma_0)^{2/3} }\right)^{\!\!\frac32-\ve}\! +  \,e^{-{1 \over 2} \Big(\! {(m|\alpha|\sigma_0)^{2/3} \over \hbar }\!\Big)^{\!\!2\ve}} \right)\,,
\end{align*}
and the proof is concluded recalling that $\psih_{\sigmazero ,\xi }(x) = \phih_{\sigmazero , x}(\xi)$. 


Next we prove claim \eqref{t2_2}. We apply the classical scattering operator, with $\beta = 2\alpha/\hbar$ to the state $\phi^{\hbar}_{\sigmazero,(\cdot)}(\xi)$. From Eq. \eqref{Sbeta}, recalling the definitions of $R_{\pm}(k)$,  we have:
\begin{align}
\big(S^{cl}_{\beta} \phi^{\hbar}_{\sigmazero,(\cdot)}\big)(\xi) =\, &\,  \phi^{\hbar}_{\sigmazero,(\cdot)}(\xi) +  {R_{-}(p/\hbar)} \big(\phi^{\hbar}_{\sigmazero,(\cdot)}(\xi) + \phi^{\hbar}_{\sigmazero,(\cdot)}(-\xi)\big)  \nonumber \\ 
=\, &\,  \psi^{\hbar}_{\sigmazero,\xi} +  {R_{-}(p/\hbar)}\big(\psi^{\hbar}_{\sigmazero,\xi} + \psi^{\hbar}_{\sigmazero,-\xi}\big)\,. \label{S-phi}
\end{align}

On the other hand, recalling the basic Identity \eqref{orto}, by simple addition and subtraction arguments and by the triangular inequality we get
\begin{equation*}
\begin{aligned}
 \big\|(\Omega_\alpha^{+})^{*} \Omega_\alpha^-\, \psih_{\sigmazero,\xi} - \big(S^{cl}_{\beta} \phi^{\hbar}_{\sigmazero,(\cdot)}\big)(\xi) \big\|_{L^2(\RE)}   
\leq\, &  \,\left\|(\Omega_\alpha^{+})^{*} \Big(\Omega_\alpha^-\, \psih_{\sigmazero,\xi} - \big(W_\beta^- \phih_{\sigmazero,(\cdot)}\big)(\xi) \Big)\right\|_{L^2(\RE)} \\ 
& + 
\big\|(\Omega_\alpha^{+})^{*} P_{ac} \big(W_\beta^-\, \phih_{\sigmazero,(\cdot)}\big)(\xi) - P_{ac} \big(S^{cl}_{\beta} \phi^{\hbar}_{\sigmazero,(\cdot)}\big)(\xi) \big\|_{L^2(\RE)}
 \\
&+ \big\|(\Omega_\alpha^{+})^{*} P_{\alpha} \big(W_\beta^- \,\phih_{\sigmazero,(\cdot)}\big)(\xi) \big\|_{L^2(\RE)}
+ \big\|P_{\alpha} \big(S^{cl}_{\beta} \phi^{\hbar}_{\sigmazero,(\cdot)}\big)(\xi) \big\|_{L^2(\RE)} \,. 
\end{aligned}
\end{equation*}

Firstly, from Identity \eqref{S-phi} and Lemma \ref{LemmaProj}, using once more the basic inequality $|R_{+}(p/\hbar)| \leq 1$ we infer
\begin{align}
\big\|P_{\alpha}\big(S^{cl}_{\beta} \phi^{\hbar}_{\sigmazero,(\cdot)}\big)(\xi) \big\|_{L^2(\RE)}
& \leq \big\|P_{\alpha} \psih_{\sigmazero,\xi}\big\|_{L^2(\RE)} + \big|R_{+}(p/\hbar)\big|\, \Big(\big\|P_{\alpha} \psih_{\sigmazero,\xi}\big\|_{L^2(\RE)} + \big\|P_{\alpha} \psih_{\sigmazero,-\xi}\big\|_{L^2(\RE)} \Big) \nonumber \\
& \leq C \Big( e^{- {m |\alpha|\,|q| \over 8 \hbar^2}} + e^{- \frac{q^2}{8 \hbar\sigmazero^2}}\Big)\,. \label{cortez1}
\end{align}

Secondly, let us notice that $\|\Omega_\alpha^+ \psi\|_{L^2(\RE)} \leq \|\psi\|_{L^2(\RE)}$,  since $\Omega_\alpha^+$  is the strong limit of operators with unit norm; thus, the same holds true for the adjoint $(\Omega_\alpha^{+})^{*}$. Hence, by arguments similar to those described above, in view of Eq. \eqref{W-phi} and  of Lemma \ref{LemmaProj}, we have
\begin{align}
\big\|(\Omega_\alpha^{+})^{*} P_{\alpha} \big(W_\beta^- \phih_{\sigmazero,(\cdot)}\big)(\xi) \big\|_{L^2(\RE)} 
& \leq \big\|P_{\alpha} \big(W_\beta^- \phih_{\sigmazero,(\cdot)}\big)(\xi)\big\|_{L^2(\RE)} \nonumber \\
& \leq \big\|P_{\alpha} \psih_{\sigmazero,\xi } \big\|_{L^2(\RE)}
+ \HE(q p) \,|R_{+}(p/\hbar)|\, \Big( \big\|P_{\alpha}\psih_{\sigmazero ,\xi }\big\|_{L^2(\RE)}
+ \big\|P_{\alpha}\psih_{\sigmazero,-\xi }\big\|_{L^2(\RE)}\Big)  \nonumber \\
& \leq C \Big( e^{- {m |\alpha|\,|q| \over 8 \hbar^2}} + e^{- \frac{q^2}{8 \hbar\sigmazero^2}}\Big)\,. \label{cortez2}
\end{align}
Again from the bound on $(\Omega_\alpha^{+})^{*}$, we infer 
\[
\big\|(\Omega_\alpha^{+})^{*} \Big(\Omega_\alpha^-\,\psih_{\sigmazero,\xi} - \big(W_\beta^-\, \phih_{\sigmazero,(\cdot)}\big)(\xi) \Big)\big\|_{L^2(\RE)} 
\leq  \big\|\Omega_\alpha^-\psih_{\sigmazero,\xi} - \big(W_\beta^- \phih_{\sigmazero,(\cdot)}\big)(\xi) \big\|_{L^2(\RE)} \,, 
\]
which is bounded by Eq. \eqref{t2_1} (proven previously).

Finally, on account of the unitarity of $\Omega_\alpha^{+}$ on $\ran(P_{ac})$ (see Remark \eqref{remWpm}), we obtain
\[
\begin{aligned}
&  \big\|(\Omega_\alpha^{+})^{*} P_{ac}\big(W_\beta^- \phih_{\sigmazero,(\cdot)}\big)(\xi) - P_{ac}\big(S^{cl}_{\beta} \phi^{\hbar}_{\sigmazero,(\cdot)}\big)(\xi) \big\|_{L^2(\RE)}  \\ 
&  = \big\| P_{ac} \big(W_\beta^- \phih_{\sigmazero,(\cdot)}\big)(\xi) - \Omega_\alpha^{+} P_{ac} \big(S^{cl}_{\beta} \phi^{\hbar}_{\sigmazero,(\cdot)}\big)(\xi) \big\|_{L^2(\RE)} \\
&  \leq  \big\| \big(W_\beta^- \phih_{\sigmazero,(\cdot)}\big)(\xi) - \Omega_\alpha^{+} \big(S^{cl}_{\beta} \phi^{\hbar}_{\sigmazero,(\cdot)}\big)(\xi) \big\|_{L^2(\RE)}  +  \big\| P_{\alpha}  \big(W_\beta^- \phih_{\sigmazero,(\cdot)}\big)(\xi) \big\|_{L^2(\RE)}  +  \big\|  \Omega_\alpha^{+} P_{\alpha}  \big(S^{cl}_{\beta} \phi^{\hbar}_{\sigmazero,(\cdot)}\big)(\xi) \big\|_{L^2(\RE)} \,.
\end{aligned}
\]
On the one hand, using once more arguments analogous to those described in the proof of the bounds \eqref{cortez1} and \eqref{cortez2}, we get
\[
\big\| P_{\alpha} \big(W_\beta^- \phih_{\sigmazero,(\cdot)}\big)(\xi) \big\|_{L^2(\RE)}  \leq 
 C \Big( e^{- {m |\alpha|\,|q| \over 8 \hbar^2}} + e^{- \frac{q^2}{8 \hbar\sigmazero^2}}\Big)
\]
and
\[
\big\|  \Omega_\alpha^{+} P_{\alpha} \big(S^{cl}_{\beta} \phi^{\hbar}_{\sigmazero,(\cdot)}\big)(\xi) \big\|_{L^2(\RE)}  \leq 
\big\| P_{\alpha}  \big(S^{cl}_{\beta} \phi^{\hbar}_{\sigmazero,(\cdot)}\big)(\xi) \big\|_{L^2(\RE)}  \leq 
  C \Big( e^{- {m |\alpha|\,|q| \over 8 \hbar^2}} + e^{- \frac{q^2}{8 \hbar\sigmazero^2}}\Big)\,.
\]
On the other hand, since $ \Omega_\alpha^{+} S^{cl}_{\beta} \phi^{\hbar}_{\sigmazero,(\cdot)}(\xi)  =  S^{cl}_{\beta}  \Omega_\alpha^{+} \phi^{\hbar}_{\sigmazero,(\cdot)}(\xi)$ (due to the fact that the operators $ \Omega_\alpha^{+}$ and $S^{cl}_{\beta}$ act on different variables), on account of Remark \ref{r:Wpm-unitary} we infer 
\[\begin{aligned}
&  \big\| \big(W_\beta^- \phih_{\sigmazero,(\cdot)}\big)(\xi) - \Omega_\alpha^{+}  \big(S^{cl}_{\beta} \phi^{\hbar}_{\sigmazero,(\cdot)}\big)(\xi) \big\|_{L^2(\RE)}   \\
& = \big\|\big(W_\beta^-\phih_{\sigmazero,(\cdot)}\big)(\xi)  - \big(S^{cl}_{\beta} \Omega_\alpha^{+} \phi^{\hbar}_{\sigmazero,(\cdot)}\big)(\xi) \big\|_{L^2(\RE)}  \\ 
& =  \big\|\big(S^{cl}_{\beta} \big(\Omega_\alpha^{+} - W_\beta^+\big)\phi^{\hbar}_{\sigmazero,(\cdot)}\big)(\xi) \big\|_{L^2(\RE)}  \\
& \leq C\left(\big\| \big(\big(\Omega_\alpha^{+} - W_\beta^+\big)\phi^{\hbar}_{\sigmazero,(\cdot)}\big)(\xi) \big\|_{L^2(\RE)} + \big\| \big(\big(\Omega_\alpha^{+} - W_\beta^+\big)\phi^{\hbar}_{\sigmazero,(\cdot)}\big)(-\xi) \big\|_{L^2(\RE)} \right)  \\ 
& \leq C\left(\left(\! {\hbar \over (m|\alpha|\sigma_0)^{2/3} }\right)^{\!\!\frac32-\ve}\!  + \,e^{-{1 \over 2} \Big(\! {(m|\alpha|\sigma_0)^{2/3} \over \hbar }\!\Big)^{\!\!2\ve}}+ e^{-  {\sigmazero^2\, p^2 \over \hbar}} + e^{- \frac{q^2}{4\hbar\sigmazero^2}}\right)\,,
 \end{aligned}
 \]
where in  the last two inequalities we used $ |\big(S^{cl}_{\beta} f\big)(q,p) | \leq C \big(|f(q,p)|+|f(-q,-p)|\big)$ (see Eq. \eqref{Sbeta} and recall that $\big|1 - {2i\,|p| \over m\,\ac}\big|^{-1}\leq 1$) and the bound in Eq. \eqref{t2_1}. 

Summing up, the above estimates imply Eq. \eqref{t2_2}. 
\end{proof}

The proof of Corollary \ref{c:2} is similar to that of Corollary \ref{c:1} and we omit it.

\end{document}